\documentclass[10pt,reqno]{amsart}

\newcommand\version{January 31, 2022}

\usepackage[ansinew]{inputenc}

\usepackage{textcomp}

\usepackage{cite}

\usepackage{amsmath,amsfonts,amsthm,amssymb,amsxtra}

\usepackage[usenames,dvipsnames]{color}
\usepackage{pst-arrow}
\usepackage{bbm}
\usepackage{fancyhdr}
\usepackage[scanall]{psfrag}
\usepackage{graphicx}
\usepackage{ifthen}
\usepackage{pstricks,pst-node,pst-plot}
\usepackage{xcolor}
\usepackage{wasysym}
\usepackage{amsthm}
\usepackage{epstopdf}
\usepackage{breqn}
\usepackage{amssymb}
\usepackage[usenames,dvipsnames]{color}
\usepackage{euscript}
\usepackage{multicol}
\usepackage{amssymb}
\usepackage{amsmath}
\usepackage{graphicx}
\usepackage[usestackEOL]{stackengine}
\usepackage{caption}
\usepackage{times} 
\usepackage{tikz}
\usepackage{mathrsfs}
\usepackage{float}
\usepackage{enumerate}
\usepackage[linktocpage]{hyperref}
\usepackage{hyperref}
\hypersetup{
    colorlinks=true,
    linkcolor=blue,
    citecolor=red,
    filecolor=magenta,      
    urlcolor=cyan,
    pdftitle={SURGERY TRANSFORMATIONS AND SPECTRAL ESTIMATES OF \texorpdfstring{$\delta'$}{Lg} BEAM OPERATORS},
    pdfpagemode=FullScreen}
\newtheorem{theorem}{Theorem}[section]
\newtheorem{proposition}[theorem]{Proposition}
\newtheorem{lemma}[theorem]{Lemma}
\newtheorem{corollary}[theorem]{Corollary}

\theoremstyle{definition}

\newtheorem{definition}[theorem]{Definition}
\newtheorem{example}[theorem]{Example}

\theoremstyle{remark}

\newtheorem{remark}[theorem]{Remark}

\numberwithin{equation}{section}

\newcommand{\B}{\mathfrak{B}}

\renewcommand{\epsilon}{\varepsilon}

\setstackgap{L}{.9\baselineskip}
\stackMath

\begin{document}

\title[Surgery Transformations--- \version]{Surgery transformations and eigenvalue estimates for quantum graphs with $\delta'$ vertex interactions }

\author{Aftab Ali and Muhammad Usman}
\address{Aftab Ali, Department of Mathematics, Lahore University of Management Sciences (LUMS), Lahore}
\email{18070008@lums.edu.pk}
\address{Muhammad Usman, Department of Mathematics, Lahore University of Management Sciences (LUMS), Lahore}
\email{m.usman08@alumni.imperial.ac.uk}


\begin{abstract} 
We extend the surgical tool box for quantum graphs to anti-standard and $\delta'$ vertex conditions. Monotonicity properties of eigenvalues of graph Laplacian with $\delta'$ interactions at vertices depend on the sign of vertex parameter. Using several interlacing inequalities between eigenvalues of graph Laplacian with diffeent vertex conditions and surgery principles we obtain new upper and lower bounds on the eigenvalues of $\delta$ and $\delta'$ Laplacians.   %graph surgery principles and record the monotonicity properties of the spectrum, keeping
%in view the possibility that vertex conditions may change within the same class after certain graph alterations. We
%also demonstrate the applications of surgery principles by obtaining several lower and upper estimates on the
%eigenvalues.
%We derive several upper bounds for the eigenvalues of Laplacian on a finite quantum graph with standard and anti standard vertex conditions in terms of the number of vertices and the total length of the graph. These bounds generalized the upper estimates on the tree graphs and are asymptotically correct i.e. they agree with the asymptotics for the eigenvalues of a quantum graph; we also provide an improved version of these estimates. For the eigenvalues of Laplacian on a finite compact metric graph with self-adjoint vertex conditions ($\delta$-type and $\delta'$-type), we also provide the collection of surgery principles.
\end{abstract}

\maketitle

\section{Introduction}
The quantum graphs - differential operators on metric graphs - is a rapidly growing branch of mathematical physics lying on the border between differential equations, spectral geometry and operator theory. During the last decade, particular attention was paid to study the impact on the spectrum of several differential operators such as Laplacian and Schr\"odinger operators on the metric graphs under the perturbations of the topology and geometry, i.e., graph surgery. These surgery principles proved to be a useful tool for the
spectral analysis of quantum graphs, in particular, for eigenvalue estimates. The usual goal is to derive spectral estimates that depend on the simple geometric properties of the graph, such as total length, the Betti number, the diameter or the number of vertices or edges of the underlying metric graph, see \cite{R,ZS,KNN,RL,A,BKKM1,KKK,F,KN,KKMM} and references therein. \\

The focus of the current work is to
\begin{itemize}
    \item extend the existing surgical tool box to the graph Laplacian with $\delta'$ and anti-standard vertex interactions and
    \item exploit this broader set of surgery principles to obtain new spectral bounds for graph Laplacians equipped with $\delta$ and $\delta'$ vertex interactions including standard, anti-standard, Neumann and Dirichlet vertex conditions.
    \end{itemize}
    Theorem 4.3 of \cite{BKKM} provides a very useful tool to derive interlacing inequalities between the eigenvalues of Laplacians with different sets of vertex conditions for a fixed underlying metric graph. Our derivation of some new eigenvalue bounds make intensive use of this result. Wherever possible we derive estimates on the general eigenvalues and not just the lowest non-zero eigenvalue.\\  
    
Some of the surgical transformations for Laplacian with $\delta'$ vertex conditions appeared in \cite{RS}.
%After introducing the notation and recollecting the necessary definitions, tools and the properties of metric graphs along with the set of permutation invariant vertex conditions in section $2$, 
We will consider three additional (or a more general form of perturbations that appeared in \cite{RS}) types of perturbations of the graphs for $\delta'$-type vertex conditions.
\begin{enumerate}[(i)]
   \item Increasing or decreasing vertex coupling parameter,
    \item a pendant graph is added to a graph, or
    \item a graph is inserted at some vertex  of other graph.
\end{enumerate}
It is well known that under these transformations, eigenvalues of Laplacian with different set of vertex interactions behave differently.
For the Kirchoff (Standard) or $\delta$ vertex conditions, increasing the $\delta$-coupling parameter at a certain vertex of the graph increases the eigenvalues. That is, a general monotonicity principle is valid. Moreover, under the perturbation (i), the \textcolor{blue}{Theorem (3.4)} of \cite{BKKM} and \textcolor{blue}{Theorem (3.1.8)} of \cite{BK1} establishes the interlacing of the eigenvalues of the two graphs. Furthermore, the Dirichelet eigenvalues of the graph are sandwiched between two consecutive eigenvalues of same underlying $\delta$-graph. For the anti-standard as well as $\delta'$ type conditions, we show that the effect on the eigenvalues under the transformation (i), in contrast to standard or $\delta$-conditions,  depends on the sign of the coupling parameters. It turns out that the perturbation (i) divides the permutation invariant conditions into two classes: those for which increasing the coupling parameter from a negative (positive) number to a negative (positive) number leads to the non-increasing eigenvalues, and those for which increasing the coupling parameter from a negative number to a positive number lead to the non-decreasing eigenvalues. Moreover, the eigenvalues of graph equipped with anti-standard conditions are squeezed between two consecutive eigenvalues of same underlying $\delta'$-graph, see, \textcolor{blue}{Theorem \eqref{strength}}. \\

The later part of section $3$ is devoted to the surgery transformation which increases the total length of the metric graph. In \textcolor{blue}{Theorem \eqref{length1}} it is shown that for permutation invariant condition such as anti-standard and $\delta'$ conditions the eigenvalues behave non-increasingly if the length of an edge is scaled by a factor greater than one. Moreover, if graph and the coupling parameters at each vertex are scaled up by some factor $t>1$ then then eigenvalues are pushed down by the factor $\frac{1}{t^2}$, see  \textcolor{blue}{Theorem \eqref{length2}}. 
 In \cite{R,BK,KMN,KN} eigenvalue estimates were derived for certain basic surgical operations of quantum graphs, namely, gluing vertices and attaching edges. The paper \cite{RS} deals with these operations for more general (permutation invarient vertex conditions) self-adjoint vertex conditions. %However, even with these operations, to date little attention has been paid to characterising the cases of equality. 
 More sophisticated surgery operations where the set of edges is changed were investigated in\cite{BKKM1,KKMM, RL}, using the symmetrisation technique
first applied to quantum graphs by L. Friedlander \cite{F} (see also \cite{KKK}). For the $\delta$-type vertex conditions it is known that the eigenvalues behave non-increasing under the surgery transformations (i) and (ii), see \cite{BKKM}.
In \textcolor{blue}{Theorem \eqref{pendant}} of the this paper, we show that for anti-standard and $\delta'$-type vertex conditions, the effect on the eigenvalues under the transformation (ii) also depend on the sign of the coupling parameters. This perturbation divides the vertex conditions into two parts: transformation (ii) that decreases the eigenvalues and those for which the perturbation (ii) increases the eigenvalues. Moreover, in \textcolor{blue}{Proposition \eqref{insert}} it is shown that for $\delta'$-type conditions the perturbation (iii) decreases the eigenvalues. In Section $4$, we provide some lower and upper bounds on the eigenvalues of the graphs equipped with (among others) $\delta$ and $\delta'$-type conditions. First, we provide some upper bounds on the lowest eigenvalue of $\delta'$-graph using some trial functions from the domain of the quadratic form. Later we use the surgical transformation to obtain a few bounds on the eigenvalues. \textcolor{blue}{Theorem \eqref{ariturk1}} provides an upper bound in terms of the betti number and the number of edges for the Neumann-Kirchoff eigenvalues and \textcolor{blue}{Theorem \eqref{delta1}} give bounds on the eigenvalues of $\delta$-graph. A similar but a better estimate can be found in \cite[\textcolor{blue}{Lemma 1.5}]{A}. However, for any flower graph the estimate in \textcolor{blue}{Theorem \eqref{ariturk1}} coincides with the estimate presented in \cite[\textcolor{blue}{Lemma 1.5}]{A}. In \textcolor{blue}{Theorem \eqref{ariturk2}} we provide the upper bounds on the eigenvalues of a metric graph with standard conditions at internal vertices and Dirichlet conditions at the vertices of degree one. For the lower and upper estimates on the eigenvalues of graph with anti-standard conditions, see \textcolor{blue}{Theorem \eqref{ec} and \eqref{anti2}}, similar estimates for standard conditions can be found in \cite{BKS}. \textcolor{blue}{Theorem \eqref{delta'} and \eqref{deltaprime}} provide lower bounds on the eigenvalues of $\delta'$-graph with negative and positive coupling parameters, respectively. The later part of this section includes a few more bounds on the eigenvalues of $\delta$ and $\delta'$-graph, these bounds are obtained using the well-known inequalities and especially the estimates on the eigenvalues of metric tree presented in \cite{R} and \cite{ZS} for standard and anti-standard vertex conditions, respectively. We also include the equivalent form of these estimates expressed in terms of other topological parameters of the underlying metric graph. In the last part, we have used a simple relation between the eigenvalues of some bipartite graphs equipped with standard and anti-standard conditions to first obtain a relation between the eigenvalues of a bipartite graph and then for general graphs equipped with  $\delta$ and $\delta'$ conditions.\\

We would like to mention that all results depend
only on the vertex conditions at those vertices which are changed by the graph transformation. At all other vertices, general self-adjoint conditions are allowed. The proofs of our results are all variational comparing Rayleigh quotients, which is a standard method in obtaining eigenvalue estimates. In fact, estimates on the quadratic form associated to the Laplace operator on suitable finite-dimensional subspaces together with an application of the min-max principle yields the desired estimates for the eigenvalues.
\section{Preliminaries}
This section provides some basics of the quantum graph. It also provides the set of interlacing inequalities between the eigenvalues of the same underlying graph equipped with different conditions such as Dirichlet, Neumann, standard, anti-standard, $\delta$ and $\delta'$. These inequalities are obtained using Theorem \cite[\textcolor{red}{Theorem 4.3}]{BKKM} and help to prove some bounds on the eigenvalues. For more detail on this subject, we refer to \cite{BK1}.\\
Let $\Gamma$ be a finite and connected graph having a finite set of edges $E$ and a finite set of vertices $V$. In addition, we identify each edge $e \in E$ with an interval $[0,\ell_e]$, where $\ell_e >0$ denote the length of the edge and let $x_e$ be coordinate along this edge which induces a metric on $\Gamma$. The sum of the lengths of each edge denotes the total length of a metric graph $\Gamma$. That is,  $$\mathcal{L}(\Gamma)=\sum\limits_{e \in E} \ell_e.$$  
Let $|E|$ and $|V|$ denote the number of edges and the number of vertices in a metric graph $\Gamma$. Each edge $e$ connects a pair of vertices $v_1$ and $v_2$, where a vertex is viewed as a subset of endpoints of edges. The vertices are disjoint such that their union coincides with the set containing each end point of all edges. An edge of $\Gamma$ can be a loop with both endpoints connected to a single vertex, and we also allow multiple edges to run between a pair of vertices. We assume that the graph $\Gamma$ is compact, which means there are a finite number of edges, each having a finite length. An edge $e$ is said to be an incident to a vertex $v$ if at least one of the endpoints belongs to that vertex and denote by $e \sim v$, and let $E_v$ denote the set of all edges incident to a vertex $v$. The number of edges incidents to a vertex $v$, denoted by $d_v$, is called the degree of a vertex. A vertex of degree one is the boundary vertex, and an edge incident to such a vertex is a boundary edge. The first Betti number, the number of independent cycles, of a metric graph $\Gamma$ is given by, 
$$\beta=|E|-|V|+1.$$
\\
Our main objective is to study the spectrum of the Laplacian $\frac{-d^2}{dx^2}$ acting on the functions living on the edges of $\Gamma$ with vertices equipped with one of the below-mentioned vertex conditions. The operator is self-adjoint in the Hilbert space 
$$L^2(\Gamma)=\bigoplus\limits_{e \in E} L^2(0,\ell_e).$$
Which consists of functions $\varphi: \Gamma \mapsto \mathbb{C}$, that are square-integrable along each edge. Let $\varphi_e$ be the restriction of $\varphi$ to some edge $e \in E$. We define this operator on a subspace of the Sobolev space $W^2_2(\Gamma \backslash V) $, consisting of functions $\varphi \in L^2(\Gamma)$, whose weak derivatives up to order two exist on each edge and are also square-integrable.  We denote these Sobolev spaces by 
$$H^2(\Gamma)=\bigoplus\limits_{e \in E} H^2(0, \ell_e).$$
equipped with the standard Sobolev norm and inner product. To describe the vertex conditions of our interest, let  us denote the limiting value of a function $\varphi$ and its derivative at any endpoint of an edge by 
$$\varphi(x_i)= \lim\limits_{x \mapsto x_i} \varphi(x) $$ and 
\begin{equation*}
\partial \varphi(x_i)
=\begin{cases}  \varphi'(x_i), \quad  \quad    \text{if $x_i$ is the left endpoint}, \\
- \varphi'(x_i),  \quad  \  \text{if $x_i$ is the right endpoint.}
\end{cases}
\end{equation*}
The vertex conditions define the linear relation between values of function and their derivatives at a vertex. The operator $L=-\frac{d^2}{dx^2}$ is a self-adjoint defined on the functions from the Sobolev space $W^2_2(\Gamma \backslash V)$ satisfying the following $\delta$-type vertex conditions.
\begin{equation} \label{delta}
\begin{cases} \varphi(x_i)= \varphi(x_j) \equiv \varphi(v_m),\quad x_i,x_j\in v_m, \quad (\text{continuity condition})\\ 
\sum\limits_{x_i \in v_m}\partial \varphi(x_i)=\alpha_m \varphi(v_m).
\end{cases}
\end{equation}
Here $\alpha_m$ are the fixed real numbers associated with each vertex $v_m$ and are commonly known as interaction strength of $\delta$-conditions. The corresponding quadratic form and its domain is given by,
\begin{equation} \label{qdelta}
h[\varphi]:=\int_\Gamma |\varphi'(x)|^2 \,dx+\sum_{m=1}^{|V|}\alpha_m|\varphi(v_m)|^2
\end{equation}
with domain, 
$$D(h)= \{ \varphi \in H^1(\Gamma): \varphi(x_i)=\varphi(x_j), \quad x_i, x_j \in v_m \}.$$
The standard conditions are defined by setting $\alpha_m=0$ in \eqref{delta} at all vertices of $\Gamma$, and the quadratic form is given by  $h[\varphi]:=\displaystyle{\int_\Gamma |\varphi'(x)|^2 dx}$   , defined on continuous functions in $H^1(\Gamma)$. We remark that the vertices of degree two equipped with standard vertex conditions do not affect the spectrum of the graph. Thus, any interior point can be regarded as a vertex of degree two equipped with standard vertex conditions and similarly, any vertex of degree two can be suppressed. \\
The Dirichlet vertex conditions are those in which the value of a function $\varphi$ is zero at each end point of an interval, that is, $$\varphi(x_i)=0, \quad x_i \in v_m. $$
Equivalently, this means the function vanishes at a vertex $v_m$ of a metric graph $\Gamma$. The Dirichlet conditions corresponds to $\delta$-type conditions if we let $\alpha_m=\infty$ in \eqref{delta}.
When Dirichlet conditions are imposed at some vertex, the functions along the incident edges do not interact (no communication between edges). Suppose all vertices are equipped with Dirichlet condition. In that case, the metric graph is the collection of independent intervals. The differential operator decouples into a direct sum of the operators on separate intervals with Dirichlet condition imposed at endpoints. The spectrum of the metric graph coincides with the disjoint union of the spectrum of each interval. That is the operator is not changed if two Dirichlet vertices are identified to form a single Dirichlet vertex. The quadratic form is given by the expression  
$h[\varphi]:=\int_\Gamma |\varphi'(x)|^2 dx$, with domain $D(h)=\{\varphi \in H^1(\Gamma): \varphi(v)=0  \}.$\\

In the $\delta'$-type condition, the role of the functions and their derivative are switched. In this case, the Laplacian is still self-adjoint on the functions from the Sobolev space $W^2_2(\Gamma \backslash V)$ satisfying the following $\delta'$-type vertex conditions.
\begin{equation} \label{delta'}
    \begin{cases}
      \partial \varphi(x_i)=\partial \varphi (x_j), \quad x_i,x_j\in v_m, \quad (\text{continuity condition}) \\
      \sum \limits_{x_i \in v_m}\varphi(x_i)=\alpha'_m \dfrac{d \varphi(v_m)}{dx}, \quad \alpha'_m \neq 0. 
    \end{cases}
\end{equation}
The corresponding quadratic form is defined on the functions in $H^1(\Gamma)$, and is given by,
\begin{equation} \label{qdelta'}
h[\varphi]:=\int_\Gamma |\varphi'(x)|^2 \,dx+\sum\limits_{m=1}^{|V|} \frac{1}{\alpha'_m} \left|\sum \limits_{x_i \in v_m}\varphi(x_i)\right|^2. 
\end{equation}
Consider the case when $\alpha'_m=0$ in \eqref{delta'} at some vertex $v_m$ of a metric graph $\Gamma$. The quadratic form in \eqref{qdelta'} only makes sense if we force that the sum of the vertex values of functions living on the incidents edges of $v_m$ is zero. The following vertex conditions are known as the anti-standard vertex conditions,
\begin{equation} \label{anti}
    \begin{cases}
      \partial \varphi(x_i)=\partial \varphi (x_j), \quad x_i,x_j\in v_m \quad (\text{continuity condition}), \\
      \sum \limits_{x_i \in v_m}\varphi(x_i)=0 \hspace{1.0in}  \quad   (\text{balance condition}). 
    \end{cases}
\end{equation}
And the quadratic form of a metric graph equipped with anti-standard vertex conditions is defined by the expression $h[\varphi]:=\int_\Gamma |\varphi'(x)|^2 \,dx$ defined on the domain. 
$$D(h)=\left \{ \varphi \in H^1(\Gamma): \sum \limits_{x_i \in v_m}\varphi(x_i)=0\right \}.$$
 Another type of vertex condition that decouples the graph into disjoints interval is the \textbf{Neumann conditions}. For these conditions, the derivatives of functions living on incident edges are zero at each vertex, and no conditions are assumed on the values of functions. That is, $$\varphi_e'(v)=0$$
Equivalently, this means that the derivatives of functions vanish at each endpoint of intervals in a metric graph $\Gamma$. The Neumann conditions corresponds to $\delta'$-type conditions if we let $\alpha'_m=\infty$ in \eqref{delta'}. The quadratic form of a metric graph equipped with such conditions is defined on the function from $H^1(\Gamma)$, given by the following expression. 
\begin{equation} \label{qneumann}
     h[\varphi]:=\int_\Gamma |\varphi'(x)|^2 \,dx.
\end{equation}
We remark that at the vertex of degree one, the anti-standard vertex conditions coincide with Dirichlet conditions, and the standard condition simplifies to Neumann conditions. Provided the length of an edge is preserved, the eigenvalues of operators corresponding to the above mentioned vertex conditions are independent of the choice of parametrization of edges in $\Gamma$. We further remark that these operators are self-adjoint and semi-bounded, and the underlying metric graph is compact, which guarantees that the spectrum of Laplacians consists of a sequence of real eigenvalues of finite multiplicity and can be listed as follows. 
$$\lambda_1 \leq \lambda_2 \leq \lambda_3 \leq \cdots,$$
 and each eigenvalue is repeated according to its multiplicity, and the corresponding eigenfunctions form an orthonormal basis of the Hilbert space $L^2(\Gamma)$. The eigenvalues and the corresponding eigenfunctions depend on the length, vertex conditions and the topological structure of the underlying metric graph \cite{BKKM,RS,BK1, KN}. For this reason, we will write $\lambda_k=\lambda_k(\Gamma)$ to reflect the spectrum's dependence on the quantum graph- a metric graph, differential operator and the vertex conditions. \\
 Generally, the explicit computation of spectrum is challenging to calculate, except in the case of a few simple
graphs, impossible to calculate explicitly. Because the relevant secular equation is transcendental to a very high degree. The zeros of these secular equations are the eigenvalues of a quantum graph, and it is not straightforward to find reliable analytical and numerical solutions to such equations. Therefore, the interesting question is whether we can find universal bounds on the eigenvalues depending on a few simple parameters of the underlying metric graph. To obtain these bounds and the corresponding extremal graphs, a better understanding of how the geometry of the graph influences the eigenvalues is required. Therefore, we attempt to understand how the eigenvalues change subject to  specific changes in the graph structure. For example, one may be interested to know if a particular graph minimizes or maximizes a given eigenvalue among all graphs with specific fixed geometric quantities such as length, number of edges, diameter, Betti number. To answer such questions, one needs to be able to make comparisons between two or more graphs. For this, we need to make some fundamental changes in the geometry of the underlying graph, such as lengthening (or shortening) an edge, gluing together vertices to obtain a new graph etc. These changes are made so that the impact on the eigenvalues are predictable; such alterations in the graph were referred to as surgery in \cite{BKKM}.\\
A standard method in obtaining eigenvalue estimates is based on variational arguments. The variational argument is based on comparing the Rayleigh quotients on suitable finite-dimensional subspaces. The Rayleigh quotient $R(\varphi)$ associated with Laplacian $L=-\dfrac{d^2}{dx^2}$ is defined as
\begin{equation}
R(\varphi)=\frac{h[\varphi]}{||\varphi||^2}
\end{equation}
on the space $L^2(\Gamma)$, the inner product is given by
\begin{equation}
<\varphi,\psi>=\int_{\Gamma} \varphi(x)\overline{\psi(x)}dx.
\end{equation}
Here, $h[\varphi]=< L \varphi, \varphi >$ is the quadratic form of the operator $L$. It is evidenced that computing the Rayleigh quotient for certain trial functions from the domain of the quadratic form leads to estimates for eigenvalues, and by minimizing the Rayleigh quotient, the eigenvalues and eigenfunctions of an operator can be obtained. Thus, the Rayleigh quotient is related to eigenvalues via the following min-max principle,  
\begin{align} \label{minmax}
        \lambda_k(\Gamma)= \min_{\underset{\dim(X) =k }{X \subset D\left(h\right) }}{\max_{0 \neq \varphi \in  X}{\frac{h[\varphi]}{||\varphi||^2}}}. 
\end{align}
If a $k$-dimensional subspace realises the minimum $\lambda_k(\Gamma)$ in \eqref{minmax}, we call it a minimising subspace for $\lambda_k(\Gamma)$.
Let $H$ and $\tilde{H}$ be two self-adjoint operators with discrete spectrum and  $h$ and $\tilde{h}$ be their corresponding semi-bounded from below quadratic forms. We say that $\tilde{h}$ is a \textit{positive rank-n perturbation} of $h$  if $\tilde{h}=h$ on some $Y\subset_n \mathrm{dom}(h)$  and either  $Y=\mathrm{dom}(\tilde{h})\subset_n \mathrm{dom(h)}$ or $\tilde{h}\geq h$ with $\mathrm{dom}(\tilde{h})=\mathrm{dom}(h)$. Here, the symbol $\subset_n$ denotes the subspace such that the quotient space $\mathrm{dom}(h)/Y$ is $n$-dimensional. \\ 

The following theorem provides the interlacing with equality characterisation between the eigenvalues of the two forms where one form is the rank $1$-perturbation of the other.  

\begin{theorem}\cite[\textcolor{red}{Theorem 4.3}]{BKKM}\label{berk1}. Let $H$ and $\tilde{H}$ be two self-adjoint operators with discrete spectrum and  $h$ and $\tilde{h}$ be their corresponding semi-bounded from below quadratic forms. If form $\tilde{h}$ is a positive rank-1 perturbation of the form $h$ then, eigenvalues of $H$ and $\tilde{H}$ satisfy
$$
\lambda_k(\Gamma)\leq \lambda_k(\tilde{\Gamma})\leq \lambda_{k+1}(\Gamma)\leq \lambda_{k+1}(\tilde{\Gamma}), \quad k\geq 1.
$$ 
Moreover, if $\lambda$ is a common eigenvalue of $H$ and $\tilde{H}$, then their respective multiplicities $m$ and $\tilde{m}$ differ by at most $1$ and the dimension of the intersection of their $\lambda$-eigenspaces is $\mathrm{min}(m,\tilde m)$. 
\end{theorem}
We will use the following notation to represent the metric graph equipped with respective vertex conditions. 
\begin{center}
\begin{tabular}{ |c|c| } 
 \hline
 Metric Graph & Vertex Conditions \\ 
 \hline
 $\Gamma^s$ & Standard  \\
 $\Gamma^a$ & Anti-Standard \\
  $\Gamma^d$ & Dirichlet  \\ 
 $\Gamma^n$ & Neumann \\
 $\Gamma^\delta$  & $\delta$ \\
 $\Gamma^{\delta'}$  & $\delta'$ \\
 \hline
\end{tabular}
\end{center}
Furthermore, if arbitrary vertex conditions are assumed at all other vertices of $\Gamma$ except $v_1$ where Dirichlet conditions are imposed,  we will use the notation $\lambda_k(\Gamma;v_1)$ to represent the eigenvalues for this graph. Moreover, if more than one vertex is equipped with Dirichlet condition, we will denote the eigenvalues by $\lambda_k(\Gamma;D)$, where $D$ is the subset of vertex set $V$ containing all Dirichlet vertices. We will only consider the case when the degree one vertex are endowed with Dirichlet conditions. \\

Let $\Gamma^n$ be the metric graph obtained from $\Gamma^a$ by imposing Neumann conditions at a vertex $v_0$. Then the corresponding quadratic form for $\Gamma^a$ and $\Gamma^n$ are  defined by the same expression \eqref{qneumann}. However, the domain $D(h^a)$ is a co-dimension one subspace of $D(h^n)$. Thus by \textcolor{blue}{Theorem} \eqref{berk1}, 
\begin{equation*}
    \lambda_k(\Gamma^n) \leq \lambda_k(\Gamma^a) \leq \lambda_{k+1}(\Gamma^n) \leq \lambda_{k+1}(\Gamma^a).
\end{equation*}
Moverover, for any Neumann graph, graph equipped with Neumann conditions, obtained from $\Gamma^a$ the domain $D(h^a)$ is a co-dimension $|V|$ subspace of $D(h^n)$, then applying the \textcolor{blue}{Theorem} \eqref{berk1} we obatined the following inequality.
\begin{equation*}
    \lambda_k(\Gamma^n) \leq \lambda_k(\Gamma^a) \leq \lambda_{k+|V|}(\Gamma^n). 
\end{equation*}
Let $\Gamma^{\delta'}$ be the graph obtained from the Neumann graph $\Gamma^n$ by imposing the $\delta'$-type conditions at vertex $v_0$ in $\Gamma^n$ with strength $\alpha'_0 >0 $, then the domains of the quadratic forms are same, $D(h^n)=D(h^{\delta'})$, and 
$$h^n[\varphi]=\int_\Gamma |\varphi'(x)|^2 \,dx \leq \int_\Gamma |\varphi'(x)|^2 \,dx+ \frac{1}{\alpha'_0} \left|\sum \limits_{x_i \in v_0}\varphi(x_i)\right|^2=h^{\delta'}[\varphi].$$
Since $D(h^a) \subset_1 D(h^n)$ and on this domain $h^n=h^{\delta'}$. Thus, $h^{\delta'}$ is a positive rank-$1$ perturbation of the form $h^n$. Therefore, 
$$\lambda_k(\Gamma^n) \leq \lambda_k(\Gamma^{\delta'}) \leq \lambda_{k+1}(\Gamma^n).$$
Similarly if $\alpha'_0 <0 $, then 
$$h^{\delta'}[\varphi]=\int_\Gamma |\varphi'(x)|^2 \,dx+ \frac{1}{\alpha'_0} \left|\sum \limits_{x_i \in v_0}\varphi(x_i)\right|^2 \leq \int_\Gamma |\varphi'(x)|^2 \,dx=h^{n}[\varphi],$$
and $D(h^a) \subset_1 D(h^{\delta'})$ and on this domain $h^n=h^{\delta'}$. Thus, $h^{n}$ is a positive rank-$1$ perturbation of the form $h^{\delta'}$. Therefore, 
$$\lambda_k(\Gamma^{\delta'}) \leq \lambda_k(\Gamma^{n}) \leq \lambda_{k+1}(\Gamma^{\delta'}).$$
An equality between a Neumann and a $\delta'$- eigenvalue is possible only if the eigenspace of
the Neumann eigenvalue contains a function which satisfy the balance condition at $v_0$, or, equivalently, the eigenspace of the $\delta'$-eigenvalue contains a function whose derivative vanishes at $v_0$.\\

%Let $\Gamma^{\delta}$ be graph obtained by specifying the $\delta$-type conditions at vertex $v_0$ of Neumann graph $\Gamma^n$ . Then, the form domains satisfy $D(h^{\delta}) \subset_1 D(h^n) $ and the respective forms are equal on $D(h^s)$. Then,   $$\lambda_k(\Gamma^n) \leq \lambda_k(\Gamma^{\delta}) \leq \lambda_{k+1}(\Gamma^n).$$
Let $\Gamma^n$ be the Neumann graph and obtain the graph $\Gamma^d$ by prescribing the Dirichlet conditions at a vertex $v_0$ in $\Gamma^n$. Since the quadratic forms of $\Gamma^n$ and $\Gamma^d$ are defined by the same expression and agree on the domain $D(h^d) \subset_1 D(h^n)$. Thus,  
$$\lambda_k(\Gamma^n) \leq \lambda_k(\Gamma^d;v_0) \leq \lambda_{k+1}(\Gamma^n).$$
For the Dirichlet graph $\Gamma^d$ and the graph  $\Gamma^a$ obtained from $\Gamma^d$ by imposing the anti-standard vertex condition at $v_0$ in $\Gamma^d$, the respective forms $h^d$ and $h^a$ coincides on $D(h^d) \subset_1 D(h^a)$. Hence, 
$$\lambda_k(\Gamma^a) \leq \lambda_k(\Gamma^d) \leq \lambda_{k+1}(\Gamma^a).$$
The equality is attained when the eigenspace of the Dirichlet eigenvalue contains a function such that its derivatives are continuous at $v_0$, or, equivalently, the eigenspace of the eigenvalue of $\Gamma^a$ contains a function that vanishes at $v_0$. \\
Moreover, for the Dirichlet graph $\Gamma^d$ and the graph  $\Gamma^{\delta'}$ obtained from $\Gamma^d$ by imposing the  $\delta'$-type condition at $v_0$ in $\Gamma^d$, the respective forms $h^d$ and $h^{\delta'}$ coincides on $D(h^d) \subset_1 D(h^{\delta'})$. Hence, 
$$\lambda_k(\Gamma^{\delta'}) \leq \lambda_k(\Gamma^d) \leq \lambda_{k+1}(\Gamma^{\delta'}).$$ 
Let $\Gamma^{\delta'}$ be finite compact metric graph with arbitrary self-adjoint vertex conditions at each vertex, except $v_0$ where $\delta'$-condition are imposed with strength $\alpha'_{0}$, and obtain a graph $\Gamma^{\delta}$ from $\Gamma^{\delta'}$ by imposing $\delta$-condition at $v_0$ with strength $\alpha_{0}$. Let $h^{\delta'}$ and $h^{\delta}$ denote the quadratic forms with domain $D(h^{\delta'})$ and $D(h^{\delta})$, respectively. Then the domain $D(h^{\delta})$ is a subspace of $D(h^{\delta'})$, and on $D(h^{\delta})$, we have 
\begin{align*}
h^{\delta}-h^{\delta'}&=\alpha_{0}|\varphi(v_0)|^2- \frac{1}{\alpha'_{0}}\left| \sum\limits_{x_i \in v_0} \varphi(x_i) \right|^2 \\
&= \alpha_{0}|\varphi(v_0)|^2- \frac{1}{\alpha'_{0}}\left| d_{v_0} \varphi(v_0) \right|^2 \\
&= |\varphi(v_0)|^2 \left(\alpha_{0}  - \frac{d_{v_0}^2}{\alpha'_{0}}\right) 
\end{align*}
Now, if $\alpha_{0}  - \frac{d_{v_0}^2}{\alpha'_{0}} \geq 0$, then $h^{\delta} \geq h^{\delta'}$.\\
Moreover, if  $\alpha_{0}  = \frac{d_{v_0}^2}{\alpha'_{0}}$, then $h^{\delta}=h^{\delta'}$ on $D(h^{\delta}) \subset_1 D(h^{\delta'})$, and by rank one perturbation the following interlacing inequalities holds 
$$\lambda_k(\Gamma^{\delta'}) \leq \lambda_k(\Gamma^{\delta}) \leq \lambda_{k+1}(\Gamma^{\delta'}). $$

For our later purpose, we provide the explicit eigenvalues of an interval $I$ and the path graph $P$.
\begin{example}
Let $I$ be a metric graph consist of single edge of length $\ell$, the eigenvalues for the operator $-\frac{d}{dx^2}$ on an interval $[0,\ell]$ are given by, 
\begin{align*}
    \lambda_k(I,v_1,v_2)&=\frac{k^2\pi^2}{\ell^2}, \quad k \geq 1, \quad  \text{Dirichlet Conditions}\\
    \lambda_{k+1}(I^s)&=\frac{k^2\pi^2}{\ell^2}, \quad k \geq 0, \quad  \text{Standard Conditions} \\
    \lambda_{k+1}(I^s)&=\frac{k^2\pi^2}{\ell^2}, \quad k \geq 0, \quad  \text{Neumann Conditions} \\
     \lambda_k(I^a)&=\frac{k^2\pi^2}{\ell^2}, \quad k \geq 1, \quad  \text{Anti-standard Conditions} 
\end{align*}

\end{example}

\begin{example}
Let $P$ be the path graph with standard vertex condition prescribed at each vertex. Due to standard vertex condition the eigenvalue problem for Laplacian $-\frac{d^2}{dx_e^2}$ on the path $P$ can be identified with eigenvalue problem on an interval of length $L(P)$ for the operator $-\frac{d^2}{dx^2}$ subject to Neumann boundary condition prescribed at both endpoints. Thus the eigenvalues of path graph $P$ are given by 
$$\lambda_{k+1}(P^s)=\frac{k^2 \pi^2}{L(P)^2}, \quad  k \geq 0.$$ 
\end{example}
\begin{example}
The eigenvalues for $-\frac{d^2}{dx_e^2}$ on the path graph $P$ of length $L(P)$ with anti standard vertex conditions imposed at each vertex coincides with the eigenvalues of $-\frac{d^2}{dx^2}$ on an interval of length $L(P)$ with Dirichlet conditions at both endpoints. Thus the eigenvalues of $P$ are given by 
$$\lambda_k(P^a)=\frac{k^2\pi^2}{L(P)^2}, \quad k \geq 1.$$

\end{example}
\section{Surgical Transformations}
In this section, we present the tools that make some fundamental changes in the geometry of the underlying metric graph, such as lengthening (or shortening) an edge, gluing together vertices; these tools will alter the given metric graph into some other graph for which the spectrum is known.  We discuss the effects of surgery principles on the spectrum of a Laplacian on a finite compact quantum graph with arbitrary self-adjoint vertex conditions ($\delta$-type, $\delta'$-type). For the metric graph equipped with $\delta$-conditions, these surgical transformations and their effect on the spectrum were studied in \cite{BKKM}. Here, we would like to investigate how some surgical operations affect the spectrum of a given metric graph $\Gamma$ and whether we can use a set of surgical operations to transform the graph into a simpler one by precisely predicting the corresponding change in the spectrum.\\
\subsection{Operation changing vertex conditions}
In this part of a section, we study the eigenvalue's dependence on the vertex conditions. The operation changing the vertex conditions does not affect the total length of a metric graph. However, it changes its connectivity in case of gluing of vertices. First, we study the dependence of eigenvalues of Laplacian on the $\delta'$-coupling parameter $\alpha'_m$. This operation of changing the $\delta$-potential at a vertex $v$ is well understood. Let $\Gamma^\delta_{\tilde{\alpha}}$ be obtained from $\Gamma^\delta_{\alpha}$ by changing the interaction strength at a vertex $v$ from $\alpha$ to $\tilde{\alpha}$. If $\infty < \alpha \leq \tilde{\alpha} \leq \infty$, then the following interlacing inequalities were established in \cite[\textcolor{blue}{ Theorem 3.4}]{BKKM} and \cite[\textcolor{blue}{Theorem 3.1.8}]{BK1}.
\begin{equation} \label{interlacing1} 
    \lambda_k(\Gamma^\delta_\alpha) \leq \lambda_k(\Gamma^\delta_{\tilde{ \alpha}}) \leq \lambda_k(\Gamma^\delta;v) \leq \lambda_{k+1}(\Gamma^\delta_\alpha)
\end{equation}
The repeated application of \eqref{interlacing1} leads to the following inequalities.
\begin{equation} \label{interlacing2} 
    \lambda_k(\Gamma^\delta_\alpha) \leq  \lambda_k(\Gamma^d) \leq \lambda_{k+|V|}(\Gamma^\delta_\alpha)
\end{equation}
Let $\Gamma$ and $\tilde{ \Gamma}$ be two quantum graphs with the same underlying metric graph but with different strengths $\alpha'_m$ in the vertex conditions. Then the domains of the quadratic forms for $\Gamma$ and $\tilde{\Gamma}$ are the same. That is, $D(h)=D(\tilde{h})$. Thus,  any $\varphi \in D(h)$ which minimizes the Rayleigh quotient for $\lambda_k(\Gamma)$ will also be in $D(\tilde{h})$ and will also minimize the Rayleigh quotient for $\lambda_k(\tilde{\Gamma})$.
The following theorem shows the spectrum's dependence on the vertex parameter $\alpha'_m$. 

\begin{theorem} \label{strength}
Let $\Gamma$ be finite, compact and connected metric graph with local, self-adjoint vertex condition at all vertices of $\Gamma$ except $v_0$, where $\delta'$-condition is imposed with strength $\alpha_0'$. Let $\tilde{\Gamma}$ be a graph obtained from $\Gamma$ by changing the strength  from $\alpha_0'$ to $\tilde{\alpha}'$, and let $\Gamma_0$ be obtained from $\Gamma$ by imposing anti-standard vertex conditions at $v_0$.
\begin{enumerate}
    \item  If   $0 < \alpha_0' < \tilde{\alpha}_0'$, or $ \alpha'_0 < \tilde{\alpha}_0' < 0$, then 
\begin{equation} \label{interlacing3}
\lambda_k(\tilde{ \Gamma}) \leq \lambda_k(\Gamma) \leq \lambda_k(\Gamma_0) \leq \lambda_{k+1}(\tilde{ \Gamma}).    
\end{equation}
\item If $\alpha_0' < 0$ and $\tilde{\alpha}_0' > 0$,  then
\begin{equation} \label{interlacing4}
\lambda_k( \Gamma) \leq \lambda_k(\tilde{\Gamma}) \leq \lambda_k(\Gamma_0) \leq \lambda_{k+1}( \Gamma).    
\end{equation}
\end{enumerate}

\end{theorem}
\begin{proof}
Let $h, \tilde{h} \ \text{and} \ h_0$ be the quadratic forms of corresponding Laplacian associated to graphs $\Gamma,\tilde{\Gamma} \ \text{and} \ \Gamma_0 $ respectively, then we have
\begin{equation*}
h[\varphi]=h_0[\varphi]+\frac{1}{\alpha_0'} \left|\sum \limits_{x_i \in v_0}\varphi(x_i)\right|^2 \text{and} \ \  \tilde{h}[\varphi]=h_0[\varphi]+\frac{1}{\tilde{\alpha}_0'} \left|\sum \limits_{x_i \in v_0}\varphi(x_i)\right|^2, 
\end{equation*}
on the appropriate sub-spaces of $H^1(\Gamma)=\bigoplus\limits_{e \in E} H^1(e)$. Since the underlying metric graph is the same and the domain of the quadratic form does not depend on strength, therefore
$$D(h)=D(\tilde{h})\ \text{and} \ D(h_0)=\left \{\varphi \in D(h): \sum \limits_{x_i \in v_0}\varphi(x_i)=0 \right \},$$
and 
\begin{equation*}
h[\varphi]-\tilde{h}[\varphi]=\left|\sum\limits_{e_{i} \in E_v} \varphi_i(v)\right|^2\left(\frac{1}{\alpha_0'}-\frac{1}{\tilde{\alpha}_0'}\right).    
\end{equation*}
Since $0 < \alpha_0' < \tilde{\alpha}_0'$, or $ \alpha'_0 < \tilde{\alpha}_0' < 0 \implies  \frac{1}{\alpha_0'}-\frac{1}{\tilde{\alpha}_0'} > 0$. Thus  $h \geq \tilde{h}$  for all $\varphi$, and the inequality follows from the min-max description of eigenvalues. Furthermore, we have
$$h[\varphi] = h_0[\varphi] \ \text{and} \ \tilde{h}[\varphi] = h_0[\varphi] \quad \text{for all} \quad \varphi \in D(h_0).$$ 
The domain $D(h)$ is larger than the domain $D(h_0)$ and the corresponding quadratic forms $h$ and $h_0$ agree on $D(h_0)$. Therefore, minimizing the Rayleigh quotient over smaller space results in larger eigenvalues. The last inequality follows from the fact that $D(h_0)$ is a subspace of $D(h)$ of co-dimension one.\\
If $\alpha_0' < 0$ and $\tilde{\alpha}_0' > 0$, then $\frac{1}{\alpha_0'} <0 $ and $\frac{1}{\tilde{ \alpha}_0'} >0  \implies \left(\frac{1}{\alpha_0'}-\frac{1}{\tilde{\alpha}_0'}\right)<0 $. Thus, $h \leq \tilde{h}$ and hence by min-max principle, $\lambda_k(\Gamma) \leq \lambda_k(\tilde{\Gamma}).$ 
\end{proof}
Equality between an eigenvalue of $\Gamma$ and $\Gamma_0$ is only possible if the eigenspace of $\lambda_k(\Gamma)$ contains an eigenfunction $\varphi $ such that at vertex $v_0$ the sum of components of $\varphi$ living on the incident edges to $v_0$ is zero. 

\begin{definition} {\textbf{(Gluing).}}
The gluing of vertices is a surgical transformation in which a graph $ \tilde{\Gamma}$ is obtained by identifying some vertices $v_1,v_2,\cdots,v_n$ of a graph $\Gamma$. If the strengths $\alpha_m$ for $m=1,2,\cdots,n$  were associated to the vertices $v_1,v_2,\cdots,v_n$, then the glued vertex $v_0$ will be specified with strength $\alpha_0=\sum\limits_{m=1}^{n} \alpha_m$.
\end{definition}
The reverse operation to gluing of vertices is named as splitting, in which a vertex $v_0$ of $\tilde{\Gamma}$ is cut into $n$ descendent vertices $v_1,v_2,\cdots,v_n$ to obtain a new graph $\Gamma$. The obtained graph $\Gamma$, using this surgery transformation, is not unique in general, as the incident edges to $v_0$ may be assigned to descendent vertices in several different ways. Furthermore, if $\varphi$ is a function on $\tilde{\Gamma}$ satisfying $\delta$-type conditions at vertex $v_0$, one can lift this function on the graph $\Gamma$ by assigning the strengths $$\alpha_m=-\dfrac{\sum\limits_{e \sim v_m}\partial \varphi_e(v_0)}{\varphi(v_0)},$$  to the descendent vertices $v_m$ for $m=1,2,\cdots, n$ in $\Gamma$. If the function $\varphi$ vanishes at $v_0$, then the descendent vertices will be endowed with Dirichlet conditions \cite{BKKM}.  
Similarly, if the graph $\tilde{\Gamma}$ is equipped with $\delta'$-type conditions then any function on $\tilde{\Gamma}$ can be lifted to $\Gamma$ by specifying the  strengths $\alpha'_m$
$$\alpha'_m=-\dfrac{\sum\limits_{e \sim v_m} \varphi_e(v_0)}{\partial \varphi(v_0)},$$  to the descendent vertices $v_m$ for $m=1,2,\cdots, n$ in $\Gamma$.\\
Let $\tilde{\Gamma}^{\delta}$ be the graph obtained by gluing two vertices of $\Gamma^{\delta}$, then the following interlacing inequalities can be easily proved using variational principles. (c.f., e.g.,  \cite[\textcolor{blue}{Theorem  3.4}]{BKKM}, \cite[ \textcolor{blue}{Theorem 3.1.10}]{BK1}).
\begin{equation}\label{dinterlac}
    \lambda_k(\Gamma^\delta) \leq \lambda_k(\tilde{ \Gamma}^\delta)  \leq \lambda_{k+1}(\Gamma^\delta) \leq \lambda_{k+1}(\tilde{ \Gamma}^\delta).
\end{equation}
For each metric graph there exist a unique graph obtained by identifying all the vertices to a single vertex, such a graph is known as flower graph. Note that, there is a one-to-one correspondence between a unique flower graph and the class of graphs with fix total number of edges, total length of metric graph and length of each edge, $(|E|,\mathcal{L},\ell_e)$. Let $\Gamma_f^{\delta}$ be the flower graph corresponding to metric graph $\Gamma^{\delta}$, then the repeated applications of \eqref{dinterlac} leads to the following inequalities.
\begin{equation}\label{dflower}
    \lambda_{k-|V|+1}(\Gamma_{f}^\delta) \leq \lambda_k(\Gamma^\delta) \leq \lambda_k( \Gamma_{f}^\delta)  \leq \lambda_{k+|V|-1}(\Gamma^\delta).
\end{equation}
The gluing operation does not depend on the sign of the strengths $\alpha_m$ for the $\delta$-graph. However, for the $\delta'$-graph we have the following theorem, which shows that the effect of gluing on the spectrum depends on the sign of the strengths $\alpha'_m$.
\begin{theorem}\cite[\textcolor{blue}{Theorem 4.2}]{RS} \label{roh}.
Let $\Gamma^{\delta'}$ be the metric graph equipped with local, arbitrary self-adjoint vertex conditions at each vertex, except $v_1$ and $v_2$, where $\delta'$-type conditions are imposed with strengths $\alpha'_1$ and $\alpha'_2$. Let $\tilde{\Gamma}^{\delta'}$ be graph obtained by gluing vertices $v_1$ and $v_2$ of $\Gamma^{\delta'}$, producing new vertex $v_0$ with strength $\alpha'_0=\alpha'_1+\alpha'_2$. Then for all $k \geq 1$ the following assertions hold.
\begin{enumerate}
    \item If $\alpha'_1, \alpha'_2 >0$ then $\lambda_k(\tilde{\Gamma}^{\delta'}) \leq \lambda_k(\Gamma^{\delta'}).$
    \item If $\alpha'_1, \alpha'_2 < 0$ then $\lambda_k(\Gamma^{\delta'}) \leq \lambda_k(\tilde{\Gamma}^{\delta'}).$
    \item If $\alpha'_1 \cdot \alpha'_2 < 0$ and $\alpha_0 >0 $ then $\lambda_k(\Gamma^{\delta'}) \leq \lambda_k(\tilde{\Gamma}^{\delta'}).$
    \item If $\alpha'_1 \cdot \alpha'_2 <0 $ and $\alpha'_0 <0 $ then $\lambda_k(\tilde{\Gamma}^{\delta'}) \leq \lambda_k(\Gamma^{\delta'}).$
    \item  If $\alpha'_1 \cdot \alpha'_2 <0 $ and $\alpha'_0 =0 $ then $\lambda_k(\Gamma^{\delta'}) \leq \lambda_k(\tilde{\Gamma}^{\delta'}).$
    \item If $\alpha'_1 \cdot \alpha'_2=0$ then $\lambda_k(\tilde{\Gamma}^{\delta'}) \leq \lambda_k(\Gamma^{\delta'}).$
\end{enumerate}
\end{theorem}
\begin{remark} \label{d'flower}
\noindent Let $\Gamma^{\delta'}$ be the metric graph and $\Gamma^{\delta'}_f$ be its corresponding flower graph. Then the following statements hold.
\begin{enumerate}
\item If $\alpha'_m > 0 $ for all $m=1,2,\cdots, |V|$ in $\Gamma^{\delta'}$, then
$$\lambda_k(\Gamma^{\delta'}_f)\leq \lambda_k(\Gamma^{\delta'}).$$ 
\item If $\alpha'_m  < 0$ for all $m=1,2,\cdots, |V|$ in $\Gamma^{\delta'}$, then
$$\lambda_k(\Gamma^{\delta'})\leq \lambda_k(\Gamma^{\delta'}_f).$$ 
\end{enumerate} 
\end{remark}

\subsection{Operations Increasing the total length}
One of the most important quantities of a metric graph is its total length, which contains information about the eigenvalues. Naturally, any change in the length of a metric graph will shift the spectrum of the respective metric graph in some direction. A more profound study on how the change in the length of the metric graph effect the spectrum can be found in \cite{EJ, RS,BKKM}. For the standard vertex condition, the scaling of a metric graph by a factor of $t>1$ reduce the eigenvalues by $t^2$. If $\mathcal{L}(\tilde{\Gamma}^s)=t \mathcal{L}(\Gamma^s)$, then 
$$\lambda_k(\tilde{ \Gamma}^s)=\frac{\lambda_k(\Gamma^s)}{t^2}.$$
Moreover, for general $\delta$-conditions, if the metric graph $\tilde{\Gamma}^{\delta}$ is obtained by scaling each edge of the graph $\Gamma^{\delta}$ by a factor of $t >0$ and coupling parameters $\alpha_m$ at each vertex are scaled with the factor $\frac{1}{t}$, then 
$$\lambda_k(\tilde{\Gamma}^{\delta})= \frac{\lambda_k(\Gamma^{\delta})}{t^2}.$$
Since the quadratic form and its domain depend on the graph's length, the following proposition show that increasing any edge's length of $\Gamma^{\delta'}$ lowers all eigenvalues.
\begin{proposition} \label{length1}
Let $\tilde{\Gamma}^{\delta'} $ be the metric graph obtained from $\Gamma^{\delta'}$ by scaling one of the edge say $e_0$ of $\Gamma^{\delta'}$ by a factor $t > 1$, then 
\begin{equation}
\lambda_k(\tilde{\Gamma}^{\delta'}) \leq \lambda_k(\Gamma^{\delta'}).    
\end{equation}
\end{proposition}
\begin{proof}
We parametrize each edge $e_j$ by $[0,\ell_j]$,
let the function $\psi(x)$ minimizes the expression \eqref{qdelta'} for the $\Gamma^{\delta'}$. Let $\varphi(x) \in D(\tilde{h}^{\delta'})$ be defined by 
\begin{equation}
 \varphi(x)=\begin{cases}
       \psi \left(\frac{x}{t}\right) \quad \text{for} \ x \in e'_0=[0,t \ell_0] \\
      \psi(x)  \quad \text{ for } \ x \in \{e\in E:e \neq e_0 \}  
    \end{cases}
\end{equation}
\begin{align*}
  \int_{e'_0}|\varphi(x)|^2dx&=\int^{tl_0}_{0}|\varphi(x)|^2dx=\int^{tl_0}_{0}\left|\psi \left(\frac{x}{t}\right)\right|^2dx=t\int_{e_0}|\psi(x)|^2dx,
  \\ &\implies \int_{e'_0}|\varphi(x)|^2dx > \int_{e_0}|\psi(x)|^2dx,\\
  \int_{e'_0}|\varphi'(x)|^2dx&=\int^{tl_0}_{0}|\varphi'(x)|^2dx=\frac{1}{t^2}\int^{tl_0}_{0} \left|\psi'\left(\frac{x}{t}\right)\right|^2dx=\frac{1}{t}\int_{e_0}|\psi'(x)|^2dx,\\
  & \implies \int_{e'_0}|\varphi'(x)|^2dx < \int_{e_0}|\psi'(x)|^2dx.
\end{align*}
We have constructed $\varphi(x)$ in such a way that the values at the vertices are unchanged $\psi(x_i)=\varphi(x_i)$, and hence 
$$\sum\limits_{m=1}^{|V|}\frac{1}{\alpha'_m} \left|\sum\limits_{x_i \in v_m} \psi(x_i)\right|^2=\sum\limits_{m=1}^{|V|}\frac{1}{\alpha'_m} \left|\sum\limits_{x_i \in v_m} \varphi(x_i)\right|^2$$
\begin{align*}
R(\psi)&=\frac{\sum\limits_{e \in E}\int_{e}|\psi'(x)|^2dx+\sum\limits_{m=1}^{|V|}\frac{1}{\alpha'_m} \left|\sum\limits_{x_i \in v_m} \psi(x_i)\right|^2}{\sum\limits_{e \in E}\int_{e}|\psi(x)|^2dx} \\ 
&=\frac{\sum\limits_{e \in E:e \neq e_0}\int_{e}|\psi'(x)|^2dx+\int_{e_0}|\psi'(x)|^2dx+\sum\limits_{m=1}^{|V|}\frac{1}{\alpha'_m} \left|\sum\limits_{x_i \in v_m} \psi(x_i)\right|^2}{\sum\limits_{e \in E:e \neq e_0}\int_{e}|\psi(x)|^2dx+\int_{e_0}|\psi(x)|^2dx} \\
&\geq \frac{\sum\limits_{e \in E:e \neq e_0}\int_{e}|\psi'(x)|^2dx+\frac{1}{t}\int_{e_0}|\psi'(x)|^2dx+\sum\limits_{m=1}^{|V|}\frac{1}{\alpha'_m} \left|\sum\limits_{x_i \in v_m} \psi(x_i)\right|^2}{\sum\limits_{e \in E:e \neq e_0}\int_{e}|\psi(x)|^2dx+t \int_{e_0}|\psi(x)|^2dx} \\
&= \frac{\sum\limits_{e \in E:e \neq e_0}\int_{e}|\varphi'(x)|^2dx+\int_{e'_0}|\varphi'(x)|^2dx+\sum\limits_{m=1}^{|V|}\frac{1}{\alpha'_m} \left|\sum\limits_{x_i \in v_m} \varphi(x_i)\right|^2}{\sum\limits_{e \in E:e \neq e_0}\int_{e}|\varphi(x)|^2dx+ \int_{e'_0}|\varphi(x)|^2dx}=R(\varphi)
\end{align*}
Since $R(\psi) \geq R(\varphi)$, so by max-min description of eigenvalues we obtain the required result.
Note that $\varphi(x) \not\in D(L^{\delta'}),$ but $\varphi(x)\in D(h^{\delta'})$.
\end{proof}
The result of the above theorem still holds when each edge $e$ of the graph $\Gamma^{\delta'}$ is scaled by a factor $t_e > 1$. This show that the spectrum of a quantum graph depends on its total length. Increasing the graph's total length also decreases the eigenvalues, so one can shift the eigenvalues in some direction by scaling the graph's total length. When $0<t<1$, the above result is reversed. The following theorem describes the relationship between the Rayleigh quotients of the graphs with total lengths $\mathcal{L}$ and $t\mathcal{L}$.   
\begin{theorem} \label{length2}
Let $\Gamma^{\delta'} $ be the metric graph with total length $\mathcal{L}(\Gamma^{\delta'})$. Obtain a graph  $ \tilde{\Gamma}^{\delta'}$ with length $t \mathcal{L}(\Gamma^{\delta'})$ by scaling each edge of $\Gamma^{\delta'}$ by $t > 0$ and scaling each strength $\alpha'_m$ by same factor, then 
\begin{equation}
   \lambda_k(\tilde{\Gamma}^{\delta'})=\frac{1}{t^2} \lambda_k(\Gamma^{\delta'}).
\end{equation}
\end{theorem}
\begin{proof}
Let $\psi(x)$ with $||\psi||_{L^2(\Gamma^{\delta'})}=1$ minimizes the Rayleigh quotient for graph $\Gamma^{\delta'}$. Then the function defined by $\varphi(x)=\frac{1}{\sqrt{t}}\psi(\frac{x}{t})$ with $||\varphi||_{L^2(\tilde{\Gamma}^{\delta'})}=1$ minimizes the Rayleigh quotient for the graph $\tilde{\Gamma}^{\delta'}$.  
Let $e_j=[0,\ell_j]$ and $\tilde{e_j}=[0,t\ell_j]$ be the parametrization of edges in  $\Gamma^{\delta'}$ and $\tilde{\Gamma}^{\delta'}$, respectively.
\begin{align*}
\tilde{R}(\varphi(\tilde{x}))&=\int_{\tilde{\Gamma}^{\delta'}}|\varphi'(\tilde{x})|^2d\tilde{x}+\sum\limits_{m=1}^{|V|}\frac{1}{\tilde{\alpha}'_m} \left|\sum\limits_{\tilde{x}_i \in \tilde{v}_m} \varphi(\tilde{x}_i)\right|^2 \\
&=\sum\limits_{j=1}^{|\tilde{E}|}\int_{\tilde{e_j}}|\varphi'(\tilde{ x})|^2d\tilde{x}+\sum\limits_{m=1}^{|V|}\frac{1}{\tilde{\alpha}'_m} \left|\sum\limits_{\tilde{x}_i \in \tilde{v}_m} \varphi(\tilde{x}_i)\right|^2  \\
%&=\sum\limits_{e}\int^{t \ell_e}_{0}|\psi'(x)|^2dx+\sum\limits_{\{v' \in V'|\beta_{v'} \neq 0\}}\frac{1}{\beta_{v'}} \left|\sum\limits_{e_i' \in E'_{v'}} \psi_i(x)\right|^2 \\
&=\sum\limits_{j=1}^{|\tilde{E}|}\int^{t \ell_j}_{0}\left|\frac{1}{\sqrt{t}}\psi'\left(\frac{\tilde{x}}{t}\right)\frac{1}{t}\right|^2d \tilde{x}+\sum\limits_{m=1}^{|V|}\frac{1}{\tilde{\alpha}'_m} \left|\sum\limits_{\tilde{x}_i \in \tilde{v}_m} \frac{1}{\sqrt t}\psi \left(\frac{\tilde{x}_i}{t} \right)\right|^2   \\
%&=\sum\limits_{i}\int^{tl_i}_{0}\frac{1}{{t^3}}\left|\varphi'\left(\frac{x}{t}\right)\right|^2dx+\frac{1}{{t}}\sum\limits_{\{v' \in V'|\beta_{v'} \neq 0\}}\frac{1}{\beta_{v'}} \left|\sum\limits_{e_i' \in E'_{v'}}  \varphi_i\left(\frac{x}{t}\right)\right|^2 \\
%&=\frac{1}{{t^3}}\sum\limits_{i}\int^{tl_i}_{0}|\varphi'\left(\frac{x}{t}\right)\frac{1}{t}|^2dx+\frac{1}{{t}}\sum\limits_{\{v' \in V'|\beta_{v'} \neq 0\}}\frac{1}{\beta_{v'}} \left|\sum\limits_{e_i' \in E'_{v'}}  \varphi_i\left(\frac{x}{t}\right)\right|^2
\end{align*}
Let $x=\frac{\tilde{x}}{t} \implies  t dx=d \tilde{x}$, and the limits of integration become $x=0 , x=\ell_e.$
\begin{align*}
%&=\frac{1}{{t^3}}\sum\limits_{i}\int^{l_i}_{0}|\varphi'(s)|^2 tds+\frac{1}{{t}}\sum\limits_{\{v \in V|\beta_{v} \neq 0\}}\frac{1}{\beta_{v}} \left|\sum\limits_{e_i \in E_{v}}  \varphi_i(s)\right|^2 \\
%&=\frac{1}{{t^2}}\sum\limits_{i}\int^{l_i}_{0}|\varphi'(s)|^2 tds+\frac{1}{{t}}\sum\limits_{\{v \in V|\beta_{v} \neq 0\}}\frac{1}{\beta_{v}} \left|\sum\limits_{e_i \in E_{v}}  \varphi_i(s)\right|^2 \\
&=\frac{1}{t^2}\left[\sum\limits_{j=1}^{|E|} \int_{e_j}|\psi'(x)|^2dx+\sum\limits_{m=1}^{|V|}\frac{1}{ {\alpha}'_m} \left|\sum\limits_{ x_i \in { v}_m}\psi(x_i)\right|^2\right] \\
%&=\frac{1}{t}\left[\frac{1}{t}\int_{\Gamma}|\varphi'(s)|^2ds+\sum\limits_{\{v \in V|\beta_{v} \neq 0\}}\frac{1}{\beta_{v}} \left|\sum\limits_{e_i \in E_{v}} \varphi_i(v)\right|^2\right]
\end{align*}
%since $s$ is a dummy variable 
%$$R(\varphi)=\frac{1}{t}\left[\frac{1}{t}\int_{\Gamma}|\varphi'(x)|^2dx+\sum\limits_{\{v \in V|\beta_{v} \neq 0\}}\frac{1}{\beta_{v}} \left|\sum\limits_{e_i \in E_{v}} \varphi_i(v)\right|^2\right].$$
The expression inside the bracket is the Rayleigh quotient for the graph $\Gamma^{\delta'}$. Therefore, by min-max description of eigenvalues.  
\begin{align*}
\lambda_k(\tilde{ \Gamma}^{\delta'})&=\min_{\underset{\dim(X) =k}{X \subset D\left(\tilde{h}\right) }}{\max_{0 \neq \varphi \in  X}} \tilde{R}(\varphi) \\
&=\min_{\underset{\dim(X) =k }{X \subset D\left(\tilde{h}\right) }}{\max_{\varphi \in  X: ||\varphi||=1}}\frac{1}{t^2}\left[\sum\limits_{j=1}^{|E|} \int_{e_j}|\psi'(x)|^2dx+\sum\limits_{m=1}^{|V|}\frac{1}{ {\alpha}'_m} \left|\sum\limits_{ x_i \in { v}_m}\psi(x_i)\right|^2\right] \\
&=\frac{1}{t^2}\left(\min_{\underset{\dim(X) =k }{X \subset D\left(h\right) }}{\max_{\psi \in  X: ||\psi||=1}}\left[\sum\limits_{j=1}^{|E|} \int_{e_j}|\psi'(x)|^2dx+\sum\limits_{m=1}^{|V|}\frac{1}{{\alpha}'_m} \left|\sum\limits_{ x_i \in { v}_m}\psi(x_i)\right|^2\right]\right)\\
&=\frac{1}{t^2} \lambda_k(\Gamma^{\delta'})
\end{align*}
If we let $t$ approaches infinity, the values of a Rayleigh quotient, and thus the eigenvalues are also pushed towards zero. And letting $t$ go to zero will force the Rayleigh quotient approaches to $\pm \infty$.  Furthermore, the result is still valid, if both the graphs are equipped anti-standard vertex conditions. 

\end{proof}

We now consider the surgical transformation that will increase the volume of a given metric graph by attaching a new subgraph to it.
\begin{definition}
{\textbf{ (Attaching a pendant graph).}} 
Given two metric graphs $\Gamma$ and $\hat{ \Gamma}$ and let $v_1 \in \Gamma$ and $w_1 \in \hat{\Gamma}$. If the metric graph $\tilde{\Gamma}$ is obtained by gluing vertices $v_1$ and $w_1$, then we speak of attaching a pendant graph $\hat{\Gamma}$ to $\Gamma$. 
\end{definition}
Let $\tilde{\Gamma}^{\delta}$ be obtained by attaching a pendant graph $\hat{\Gamma}^{\delta}$ to $\Gamma^{\delta}$, then it was proved in \cite[\textcolor{blue}{Theorem 3.10}]{BKKM} that if for some $r$ and $k$, $\lambda_r(\hat{\Gamma}^{\delta}) \leq \lambda_k(\Gamma^{\delta})$ then $$\lambda_{k+r-1}(\tilde{\Gamma}^{\delta}) \leq \lambda_k(\Gamma^{\delta}).$$
For the graphs equipped with $\delta'$-type condition it was established in \cite[\textcolor{blue}{Theorem 3.5}]{RS} that attaching a pendant edge to some vertex of $\Gamma^{\delta'}$ to obtained a graph $\tilde{\Gamma}^{\delta'}$ leads to the following inequality, 
$$\lambda_{k}(\tilde{\Gamma}^{\delta}) \leq \lambda_k(\Gamma^{\delta}).$$
This inequality also hold when the graphs are equipped with anti-standard vertex conditions. Moreover, the repeated application of this result also shows that the same result is true when the graph $\tilde{\Gamma}^{\delta'}$ is formed by attaching a graph $\hat{\Gamma}^{\delta}$ to $\Gamma^{\delta'}$.
\begin{theorem} \label{pendant}
Let $\Gamma$ be a finite, compact and connected metric graph, let $v_1$ be a vertex of $\Gamma$ and let $L=-\frac{d^2}{dx^2}$ be the Laplacian in $L^2(\Gamma)$ subject to arbitrary local, self-adjoint vertex conditions at each vertex $v_m \in V \setminus\{v_1\}$, with $\delta'$ vertex condition at $v_1$ , with strength $\alpha'_1$. Let $\hat{\Gamma}$ be a finite, compact and connected metric graph, with arbitrary local, self-adjoint vertex conditions at each vertex $w_m \in \hat{V} \setminus \{w_1\}$, with $\delta'$ vertex condition at $w_1$, with strength $\hat{ \alpha}'_2$. Let $\tilde{\Gamma}$ be graph formed by attaching the pendant graph $\hat{\Gamma}$ to $\Gamma$, if for some $r,k_0$, we have, $$\lambda_r(\hat{\Gamma})\leq \lambda_{k_0}(\Gamma),$$ then the following assertions hold.
\begin{enumerate}
    \item  If $\alpha'_1, \hat{\alpha}'_2 > 0$, then
    $${\lambda_{k+r}(\tilde{\Gamma}) \leq \lambda_k(\Gamma), \quad k \geq k_0}.$$
    
    \item If $\alpha'_1 \cdot \hat{\alpha}'_2 = 0$ then, $${\lambda_{k+r}(\tilde{\Gamma}) \leq \lambda_k(\Gamma), \quad k \geq k_0}.$$
    
    \item If $\alpha'_1 \cdot \hat{\alpha}'_2 < 0 \ \text{and} \  \alpha'_1+ \hat{\alpha}'_2 < 0$ then,
    $${\lambda_{k+r}(\tilde{\Gamma}) \leq \lambda_k(\Gamma), \quad k \geq k_0}.$$
    
    \item If $\alpha'_1, \hat{\alpha}'_2 < 0$, then
    $${\lambda_{k+r}(\tilde{\Gamma}) \geq \lambda_k(\Gamma), \quad k \geq k_0}.$$
    
    \item If $\alpha'_1 \cdot \hat{\alpha}'_2 < 0 \ \text{and} \ \alpha'_1+ \hat{\alpha}'_2 > 0$ then
    $${\lambda_{k+r}(\tilde{\Gamma}) \geq \lambda_k(\Gamma), \quad k \geq k_0}.$$
    
    \item  If $\alpha'_1 \cdot \hat{\alpha}'_2 < 0 \ \text{and}\ \alpha_1+ \hat{\alpha}'_2 = 0$ then
    $${\lambda_{k+r}(\tilde{\Gamma}) \geq \lambda_k(\Gamma), \quad k \geq k_0}.$$

\end{enumerate}
\end{theorem}
\begin{proof}
\begin{enumerate}
    \item 
 Let $\Gamma'$ be the disconnected graph whose components are $\Gamma$ and $\hat{\Gamma}$, that is, $\Gamma'=\Gamma \cup \hat{\Gamma}$. The assumption $\lambda_r(\hat{\Gamma}) \leq \lambda_{k_0}(\Gamma)$, and the fact that the spectrum of $\Gamma'$ is equal to the union of the spectra of $\Gamma$ and $\hat{\Gamma}$ implies that every $k$-th eigenvalue of $\Gamma$ is some $m$-th eigenvalue of $\Gamma'$, 
$$\lambda_{k_0}(\Gamma)=\lambda_m(\Gamma')$$ for some $m=m(r,k_0) \geq r+k_0$. Now attach the vertex $w_1$ of $\hat{\Gamma}$ with a vertex $v_1$ of $\Gamma$ to form $\tilde{\Gamma}$, then by \textcolor{blue}{Theorem} \eqref{roh}, when $\alpha'_1,\hat{\alpha}'_2 > 0$,  we obtain   $$\lambda_m(\tilde{\Gamma}) \leq \lambda_m(\Gamma '),$$
combining this inequality with the estimate $m \geq r+k_0$, we obtain 
$$\lambda_{r+k}(\tilde{\Gamma}) \leq \lambda_k(\Gamma), \quad  k \geq k_0.$$
    (4) Let $\Gamma'$ be the disconnected graph whose components are $\Gamma$ and $\hat{\Gamma}$, that is, $\Gamma'=\Gamma \cup \hat{\Gamma}$. The assumption $\lambda_r(\hat{\Gamma}) \leq \lambda_{k_0}(\Gamma)$, and the fact that the spectrum of $\Gamma'$ is equal to the union of the spectra of $\Gamma$ and $\hat{\Gamma}$ implies that 
$$\lambda_{k_0}(\Gamma)=\lambda_m(\Gamma')$$ for some $m=m(r,k_0) \geq r+k_0$. Now attach the vertex $w_1$ of $\hat{\Gamma}$ with a vertex $v_1$ of $\Gamma$ to form $\tilde{\Gamma}$, then by \textcolor{blue}{Theorem} \eqref{roh}, when $\alpha'_1, \hat{\alpha}'_2 < 0$,  we obtain   $$\lambda_m(\tilde{\Gamma}) \geq \lambda_m(\Gamma '),$$
combining this inequality with the estimate $m \geq r+k_0$, we obtain 
$$\lambda_{r+k}(\tilde{\Gamma}) \geq \lambda_k(\Gamma), \quad  k \geq k_0.$$
   
\end{enumerate}
The other parts of the theorem can be proved correspondingly.
\end{proof}

\begin{definition}
The inserting a graph at vertex is a surgical operation in which a metric graph  $\tilde{\Gamma}$ is obtained by removing a vertex $v_0$ of $\Gamma$ and attaching its incident edges to some vertices of $\hat{\Gamma}$.  
Let $e_1,e_2,\cdots,e_n$ be the incident edges to a vertex $v_0$ of a metric graph $\Gamma$ and $w_1,w_2,\cdots,w_m, m \leq n$ be the vertices in $\hat{\Gamma}$ to which an edge has been so attached. If $\alpha_0$ and $\hat{\alpha}_1,\hat{\alpha}_2, \cdots, \hat{\alpha}_m$ are the coupling parameters associated to the vertices $v_0$ and $w_1,w_2,\cdots,w_m$, respectively. Then the coupling parameters should be placed at the vertices of $\tilde{w}_1,\tilde{w}_2,\cdots,\tilde{w}_m,$ of $\tilde{\Gamma}$ in such a way that their sum is equal to $\alpha_0 +\hat{\alpha}_1+\hat{\alpha}_2+ \cdots+ \hat{\alpha}_m$. 
\end{definition}

A similar result to the following theorem was proved in \cite[\textcolor{blue}{Theorem $3.10$}]{BKKM}, , in which it was assumed that if the vertices of the graphs $\Gamma$ and $\hat{\Gamma}$ are equipped with $\delta$-type and standard vertex conditions, respectively. Then, for all $k$ such that $\lambda_k(\Gamma) \geq 0$, $$\lambda_k(\tilde{\Gamma}) \leq \lambda_k(\Gamma).$$       
The reason for imposing standard conditions on $\hat{\Gamma}$ is that we want that the sum of the strengths assigned to the new vertices of $ \tilde{\Gamma}$ must equal to the strength assigned at $v_0$ of $\Gamma$. The following result omits this restriction. Therefore, the eigenvalues of $\tilde{\Gamma}$ are bounded above by the eigenvalues of the graphs with different strengths at $v_0$ each strength equal to the sum of the strengths of vertices in $\tilde{\Gamma}$ to which the incident edges of $v_0$ are attached.  Moreover, we believe that the following theorem gives an improved result when $r \geq m$. 
\begin{theorem}
 Let $\tilde{\Gamma}^{\delta}$ be obtained by inserting vertices   $w_1,w_2,\cdots,w_m$ of $\hat{\Gamma}^{\delta}$ at vertex $v_0$ of $\Gamma^{\delta}$. If for some $r\geq m$ and $k$, we have, $ \lambda_r(\hat{\Gamma}^{\delta}) \leq \lambda_{k}(\Gamma^{\delta}).$ Then, 
\begin{equation}
\lambda_{k+r-m}(\tilde{\Gamma}^{\delta})\leq \lambda_k(\Gamma^{\delta}).     
\end{equation}
\end{theorem}
\begin{proof}
Let $\check{\Gamma}^{\delta}$ be a graph obtained by gluing the vertices $w_1,w_2,\cdots,w_m$ of $\hat{\Gamma}^{\delta}$ to form a single vertex $w^*$, then
$$\lambda_{r-m+1}(\check{\Gamma}^{\delta}) \leq \lambda_r(\hat{\Gamma}^{\delta}). $$
The assumption $ \lambda_r(\hat{\Gamma}^{\delta}) \leq \lambda_{k}(\Gamma^{\delta}),$ implies that $\lambda_{r-m+1}(\check{\Gamma}^{\delta}) \leq \lambda_{k}(\Gamma^{\delta})$.
Obtain the graph $\check{\check{\Gamma}}^{\delta}$ by attaching the vertex $w^*$ of pendant graph $\check{\Gamma}^{\delta}$ at vertex $v_0$ of $\Gamma^{\delta}$, this implies,  
$$\lambda_{k+r-m}(\check{\check{\Gamma}}^{\delta}) \leq \lambda_k(\Gamma^{\delta}).$$
Now, cut through the vertex $v_0$ of $\check{\check{\Gamma}}^{\delta}$ to restore the vertices $\tilde{w}_1,\tilde{w}_2,\cdots,\tilde{w}_m$ to form a graph $\tilde{\Gamma}^{\delta}$, thus
 $$\lambda_{k+r-m}(\tilde{\Gamma}^{\delta})\leq \lambda_{k+r-m}(\check{\check{\Gamma}}^{\delta}).$$ 
\end{proof}
\begin{corollary}
Let $\Gamma$ be a finite, compact and connected metric graph with arbitrary self adjoint vertex conditions at each vertex. Let $\tilde \Gamma$ be a metric graph obtained by lengthening an edge of $\Gamma$.
If for some $k_0$, $\lambda_{k_0}(\Gamma)\geq 0$, then 
\begin{equation}
    \lambda_k(\tilde \Gamma) \leq \lambda_k(\Gamma), \quad k \geq k_0.
\end{equation}
\end{corollary}
\begin{proof}
Let $\Gamma^s$ be an edge of finite length with endpoints equipped with standard vertex conditions,  and $\mathcal{C}^s$ be a loop obtained by joining the boundary points of $\Gamma^s$. Consider any interior point $x^*$ on the edge of $\Gamma$ as a vertex of degree two equipped with standard condition. 
Since 
\begin{equation*}
    0=\lambda_1(\mathcal{C}^s) \leq \lambda_{k_0}(\Gamma)
\end{equation*}
Obtain a graph $\hat{\Gamma}$ by attaching $\mathcal{C}^s$ as a pendant graph to $\Gamma$ at an interior point $x^*$ then.
\begin{equation*}
    \lambda_k (\hat{\Gamma}) \leq \lambda_k(\Gamma), \quad k \geq k_0.
\end{equation*}
Now restore the endpoints of $\Gamma^s$ in accordance with the definition of splitting of the vertex to obtained a graph $\tilde \Gamma$, having two interior points as a vertex of degree two each equipped with standard conditions. Which can easily be removed and considered as interior points of an edge of $\Gamma$. Thus,
\begin{equation*}
    \lambda_k (\tilde{\Gamma} ) \leq \lambda_k(\hat{\Gamma}), \quad k \geq k_0.
\end{equation*}
\end{proof}
\begin{proposition} \label{insert}
Suppose $\tilde{\Gamma}$ is formed by inserting a graph $\Gamma^a$ at a vertex $v_0$ of $\Gamma$ in a way that incident edges of $v_0$ are attached to only two vertices , say $w_1$ and $w_2$, of $\Gamma^a$. Assume that arbitrary local self-adjoint vertex conditions are imposed on each vertex $v_m$ of $\Gamma$ except $v_0$, where $\delta'$-condition is imposed with strength $\alpha'_0$, and anti-standard conditions are prescribed at each vertex of $\Gamma^a$ prior to insertion. If for some $k_0$, $\lambda_{k_0}(\Gamma)\geq 0$, then
\begin{equation*}
    \lambda_{k+1}(\tilde{\Gamma}) \leq \lambda_k(\Gamma), \quad k\geq k_0.
\end{equation*}
\end{proposition}
\begin{proof}
Let $\hat{\Gamma}^a$ be a graph obtained by gluing the vertices $w_1$ and $w_2$ of $\Gamma^a$ to form a single vertex $w^*$ with total strength zero.  The assumption $ \lambda_{k}(\Gamma) \geq 0$, for all  $k \geq k_0$, along with the fact that the spectrum of $\hat{\Gamma}^a$ is non-negative,  implies that $\lambda_1(\hat{\Gamma}^a)=0 \leq \lambda_k(\Gamma)$, for all $k \geq k_0$.
Let $\check{\Gamma}$ be obtained by  attaching the vertex $w^*$ of pendant graph $\hat{\Gamma}^a$ at vertex $v_0$ of $\Gamma$, then by \textcolor{blue}{Theorem \eqref{pendant}(2)}, 
$$\lambda_{k+1}(\check{\Gamma}) \leq \lambda_k(\Gamma), \quad k \geq k_0.$$
Now, cut through the vertex $v_0$ of $\check{\Gamma}$ to restore the vertices $w_1$ and $w_2$ to obtain the graph $\tilde{\Gamma}$, the restoration of vertices is performed in such a way that the sum of strengths $\tilde{\alpha}'_1$, $\tilde{\alpha}'_2$, assigned to $\tilde{w}_1$ and $\tilde{w}_2$, is $\alpha'_0$. Moreover, the strengths are assigned such that the the operation of splitting the vertex $v_0$ lowers all eigenvalues of $\check{\Gamma}$. Thus,
 $$\lambda_{k+1}(\tilde{\Gamma})\leq  \lambda_{k+1}(\check{\Gamma}).$$ 
\end{proof}

\section{Eigenvalue Estimates}
This section provides upper and lower estimates for Laplacian on a metric graph equipped with  self-adjoint vertex conditions ($\delta$ and $\delta'$-type conditions).  This section is divided into two subsections.

\subsection{Estimates on the eigenvalues of \texorpdfstring{$\delta$}{Lg}-graph} In this part of the section we provide some upper and lower estimates on the eigenvalues of the $\delta$-graph. 
\begin{theorem} \label{ariturk1}
Let $\Gamma^s$ be finite, compact connected metric graph with standard conditions, then 
\begin{equation}\label{est1}
    \lambda_k(\Gamma^s) \leq \left( \frac{(k-1+\beta+|E|)\pi}{\mathcal{L}(\Gamma^s)}\right)^2.
\end{equation}
\end{theorem}
\begin{proof}
Consider an interval $I^s$ of length $\mathcal{L}(\Gamma^s)$ equipped with standard vertex conditions. Obtain a path graph $P^s$ by creating vertices of degree  two at interior points of $I^s$ such that the number of edges and their lengths are same as in $\Gamma^s$. Cut through the degree two vertices of $P^s$ to form a new graph $\hat{\Gamma}^s$ which consists of $|E|$ disconnected edges. Create a metric tree $T^s$ by gluing the vertices of $\hat{\Gamma}^s$ in such a way that the graph $\Gamma^s$ can be obtained from $\Gamma^s$ by pairwise gluing of $\beta$ pair of vertices of $T^s$. Thus,
\begin{equation} 
    \left(\frac{k \pi}{\mathcal{L}(\Gamma^s)}\right)^2=\lambda_k(I^s)=\lambda_k(P^s) \geq \lambda_k(\hat{\Gamma}^s) \geq \lambda_{k-|E|+1}(T^s) \geq \lambda_{k-|E|+1-\beta}(\Gamma^s).
\end{equation}
Equivalently, this estimates can be written as, 
$$\lambda_k(\Gamma^s) \leq \left( \frac{(k-2+2\beta+|V|)\pi}{\mathcal{L}(\Gamma^s)}\right)^2.$$
\end{proof}
We can also obtain a lower bound on $\lambda_k(\Gamma^s)$ as follows. However, this bound is not a better estimate as it only makes sense when $k\geq |E|-1$.
$$\left(\frac{(k-|E|+1) \pi}{\mathcal{L}(\Gamma^s)} \right)^2=\lambda_{k-|E|+1}(I)=\lambda_{k-|E|+1}(P) \leq \lambda_k(\hat{\Gamma}^s) \leq \lambda_k(T^s) \leq  \lambda_k(\Gamma^s).$$
We use this estimate to obtain a few bounds on the eigenvalues of $\Gamma^{\delta}$. Let $\alpha=\sum\limits_{m=1}^{|V|} \alpha_m$ denote the total strength of $\Gamma^{\delta}$ and $\Gamma^{\delta}_f$ be the corresponding flower graph. Let $\Gamma^s_f$ be the same underlying flower graph equipped with standard conditions. If $\alpha \leq 0$, then the above result implies that $$\lambda_k(\Gamma^{\delta}) \leq \lambda_k(\Gamma^{\delta}_f) \leq  \lambda_k(\Gamma^s_f) \leq  \left( \frac{(k-1+2|E|)\pi}{\mathcal{L}(\Gamma^{\delta})}\right)^2.$$
Similarly, if $\alpha > 0$, then 
$$\lambda_k(\Gamma^{\delta}) \leq \lambda_k(\Gamma^{\delta}_f) \leq  \lambda_{k+1}(\Gamma^s_f) \leq  \left( \frac{(k+2|E|)\pi}{\mathcal{L}(\Gamma^{\delta})}\right)^2.$$
\begin{theorem} \label{delta1}
Let $\Gamma^{\delta}$ be a metric graph with $|E|$ number of edges and $|V|$ number of vertices. Then the following assertions hold. 
\begin{enumerate}
    \item If $\alpha=0$, then 
    \begin{equation*}
        \left( \frac{(k-|V|+1)\pi}{\mathcal{L}(\Gamma^{\delta})}\right)^2  \leq    \lambda_k(\Gamma^{\delta}) \leq \left( \frac{(k+|E|)\pi}{\mathcal{L}(\Gamma^{\delta})}\right)^2.
    \end{equation*}
    \item If $\alpha > 0 $, then 
    \begin{equation}
   \left( \frac{(k-|V|+1)\pi}{\mathcal{L}(\Gamma^{\delta})}\right)^2  \leq    \lambda_k(\Gamma^{\delta}) \leq \left( \frac{(k+|E|+1)\pi}{\mathcal{L}(\Gamma^{\delta})}\right)^2.
    \end{equation}
    \item If $\alpha < 0$, then 
    \begin{equation}
        \left( \frac{(k-|V|)\pi}{\mathcal{L}(\Gamma^{\delta})}\right)^2  \leq    \lambda_k(\Gamma^{\delta}) \leq \left( \frac{(k+|E|)\pi}{\mathcal{L}(\Gamma^{\delta})}\right)^2.
    \end{equation}
\end{enumerate}
\end{theorem}
\begin{proof}
Let $I^s$ be an interval of length $\mathcal{L}(\Gamma^{\delta})$ equipped with standard vertex conditions and let $P^s$ be a path graph obtained by creating vertices of degree two at interior points of $I^s$ such that the number of edges and their lengths are same as in $\Gamma^{\delta}$. Then, 
$$\lambda_k(I^s)=\lambda_k(P^s).$$
Obtain a flower graph $P^s_f$ from $P^s$, then \eqref{dinterlac} implies that 
$$\lambda_{k-|E|}(P^s_f) \leq \lambda_{k}(P^s) \leq \lambda_k(P^s_f) \leq \lambda_{k+|E|}(P^s).$$
Thus, 
\begin{equation*}
    \left( \frac{k\pi}{\mathcal{L}(\Gamma^{\delta})}\right)^2=\lambda_k(I^s)=\lambda_k(P^s)\leq \lambda_k(P^s_f) \leq \lambda_{k+|E|}(P^s)=\lambda_{k+|E|}(I^s)=    \left(\frac{(k+|E|)\pi}{\mathcal{L}(\Gamma^{\delta})}\right)^2.
\end{equation*}
\begin{enumerate}
    \item If $\alpha=0$, then 
    $$  \left( \frac{(k-|V|+1)\pi}{\mathcal{L}(\Gamma^{\delta})}\right)^2  \leq \lambda_{k-|V|+1}(P^s_f) \leq \lambda_k(\Gamma^{\delta})\leq \lambda_k(P^s_f) \leq \left( \frac{(k+|E|)\pi}{\mathcal{L}(\Gamma^{\delta})}\right)^2.$$
    \item If $\alpha > 0$, then 
    $$  \left( \frac{(k-|V|+1)\pi}{\mathcal{L}(\Gamma^{\delta})}\right)^2  \leq \lambda_{k-|V|+1}(P^s_f)     \leq \lambda_k(\Gamma^{\delta})\leq \lambda_{k+1}(P^s_f) \leq \left( \frac{(k+|E|+1)\pi}{\mathcal{L}(\Gamma^{\delta})}\right)^2.$$
    \item If $\alpha < 0$, then 
    $$  \left( \frac{(k-|V|)\pi}{\mathcal{L}(\Gamma^{\delta})}\right)^2  \leq \lambda_{k-|V|}(P^s_f)     \leq \lambda_k(\Gamma^{\delta})\leq \lambda_k(P^s_f) \leq \left( \frac{(k+|E|)\pi}{\mathcal{L}(\Gamma^{\delta})}\right)^2.$$
\end{enumerate}
\end{proof}
\begin{theorem} \label{ariturk2}
Let $\Gamma$ be finite compact connected metric graph with pendant vertices equipped with Dirichlet conditions and all other internal vertices the standard conditions are imposed. Let $|D|$ and $|S|$ denote the number of Dirichlet and number of standard vertices in $\Gamma$. Then \begin{equation}
    \lambda_k(\Gamma) \leq \left( \frac{(k-2+2\beta+2|D|+|S|)\pi}{\mathcal{L}(\Gamma)}\right)^2.
\end{equation}
\end{theorem}
\begin{proof}
Let $\Gamma^s$ be same underlying metric graph as $\Gamma$, but equipped with standard vertex conditions. Then by \eqref{est1}
$$\lambda_k(\Gamma^s) \leq \left( \frac{(k-1+\beta+|E|)\pi}{\mathcal{L}(\Gamma^s)}\right)^2.$$
Obtain the graph $\Gamma$ by imposing Dirichlet conditions at required vertices of $\Gamma^s$. Then by \eqref{interlacing1}
$$\lambda_{k-|D|}(\Gamma) \leq \lambda_k(\Gamma^s) \leq \lambda_k(\Gamma)\leq \lambda_{k+|D|}(\Gamma^s).$$
$$\lambda_k(\Gamma) \leq \left( \frac{(k-2+2\beta+|D|+|V|)\pi}{\mathcal{L}(\Gamma^s)}\right)^2. $$
Equivalently, this estimates can be written as, 
$$\lambda_k(\Gamma) \leq \left( \frac{(k-2+2\beta+2|D|+|S|)\pi}{\mathcal{L}(\Gamma)}\right)^2.$$
\end{proof}
For the metric graphs in \textcolor{blue}{Theorem} \eqref{ariturk1} and \textcolor{blue}{Theorem} \eqref{ariturk2},  similar but improved estimates on the eigenvalues can be found in \cite[\textcolor{blue}{Lemma 1.5}]{A}. However, for any flower graph equipped with standard conditions the estimate \eqref{est1} coincides with the estimate presented in \cite[\textcolor{blue}{Lemma 1.5}]{A}.
\subsection{Estimates on the eigenvalues of \texorpdfstring{$\delta'$}{Lg}-graph}
We establish a few estimates on the lowest eigenvalue of the $\delta'$-graph using some trials functions from the domain of the quadratic form and also on general eigenvalues using some established inequalities.  \\
Since the domain of the quadratic form of the $\delta'$-graph $\Gamma^{\delta'}$ contain a constant function, in particular it contains $\hat \varphi(x)=\frac{1}{\sqrt{\mathcal{L}(\Gamma^{\delta'})}}$. Therefore, 
\begin{align*}
 \lambda_1(\Gamma^{\delta'})&= \underset{\varphi \in D(h^{\delta'})}{\min} h^{\delta'}[\varphi] \\
 &=\underset{\varphi \in D(h^{\delta'})}{\min} \left( \int_{\Gamma^{\delta'}}|\varphi'(x)|^2 dx +\sum\limits_{m=1}^{|V|} \frac{1}{\alpha'_m} \left|\sum \limits_{x_i \in v_m}\varphi(x_i)\right|^2 \right)  \\
 &\leq  \int_{\Gamma^{\delta'}}|\hat \varphi'(x)|^2 dx +\sum\limits_{m=1}^{|V|} \frac{1}{\alpha'_m} \left|\sum \limits_{x_i \in v_m}\hat \varphi(x_i)\right|^2 \\
 &= \sum\limits_{m=1}^{|V|} \frac{1}{\alpha'_m} \left|\sum \limits_{x_i \in v_m} \frac{1}{\sqrt{\mathcal{L}(\Gamma^{\delta'})}} \right|^2 \\
 &=\frac{1}{\mathcal{L}(\Gamma^{\delta'})} \sum\limits_{m=1}^{|V|} \frac{1}{\alpha'_m} \left|\sum \limits_{x_i \in v_m} 1 \right|^2 \\
 &=\frac{1}{\mathcal{L}(\Gamma^{\delta'})} \sum\limits_{m=1}^{|V|} \frac{1}{\alpha'_m} |d_{v_m}|^2 \\
 &=\frac{1}{\mathcal{L}(\Gamma^{\delta'})} \sum\limits_{m=1}^{|V|} \frac{1}{\alpha'_m} |d_{v_m}|^2
 \end{align*}
\begin{remark}
Even if there are some vertices where $\alpha'_m=0$, we can still use this upper bound and in this case the summation is taken over the vertices where strengths $\alpha'_m$ are non zero.
\begin{equation*}
    \lambda_1(\Gamma^{\delta'}) \leq \frac{1}{\mathcal{L}(\Gamma^{\delta'})} \left( \sum\limits_{\{v \in V| \alpha'_v \neq 0\}} \frac{d^2_{v}}{\alpha'_v} \right)
\end{equation*}
\end{remark}

\begin{enumerate}
   \item If the degree of each vertex is $n$ then, $$\lambda_1(\Gamma^{\delta'}) \leq \frac{n^2} { \mathcal{L}(\Gamma^{\delta'})} \sum\limits^{|V|}_{m=1} \frac{1}{\alpha'_m}.$$
    
    \item If the strength at each vertex is same, say $\alpha'$, then 
    $$\lambda_1(\Gamma^{\delta'}) \leq \frac{1}{\alpha' \mathcal{L}(\Gamma^{\delta'})} \sum\limits^{|V|}_{m=1} |d_{v_m}|^2.$$
    
    \item In addition, if the total number of vertices are $|V|$ and degree of each vertex is $n$ then, $$\lambda_1(\Gamma^{\delta'}) \leq \frac{n^2 \ |V|}{\alpha' \mathcal{L}(\Gamma^{\delta'})}.$$
    \item If the strengths are evenly distributed, then
    $$\lambda_1(\Gamma^{\delta'}) \leq \frac{n^2 \ |V|^2 }{\alpha' \mathcal{L}(\Gamma^{\delta'})}.$$
    \end{enumerate}
    The question now is how close one can push the lowest eigenvalue towards this upper bound. The quadratic form does not converge very close to lowest eigenvalue when the strengths are evenly distributed. So, evenly distribution does not increase the eigenvalue. How can we get better estimation?   \\
Total degree of a graph is defined as the sum of degree of each vertex. Handshaking theorem states that the degree of a graph $\Gamma$ equals twice the number of edges in $\Gamma$ and the total degree of a graph is even.
\begin{equation*}
    degree \ of \  a \  graph = \sum_{v \in V} deg(v)=2 \ (number \ of \  edges \ in \ \Gamma)
\end{equation*}

\begin{remark}
We express the upper bound on the lowest eigenvalue in terms of total degree of a graph and total number of edges. Although, the following is not a good approximation as compare to earlier one.
Let $|V|$ and $|E|$ denote the total number of vertices and total number of edges in a graph $\Gamma$, respectively. Since
\begin{align*}
    |d_{v_1}|^2+|d_{v_2}|^2+\cdots+|d_{v_{|V|}}|^2 &\leq (|d_{v_1}|+|d_{v_2}|+\cdots+|d_{v_{|V|}}|)^2 \\
    &=[deg(\Gamma)]^2 \\
    &=4 \ |E|^2.
\end{align*}
If the strengths are evenly distributed, then 
\begin{equation*}
    \lambda_1(\Gamma^{\delta'}) \leq \frac{|V|}{\beta \mathcal{L}(\Gamma^{\delta'})} \ [deg(\Gamma^{\delta'})]^2,
\end{equation*}
and in terms of number of edges,
\begin{equation*}
    \lambda_1(\Gamma^{\delta'}) \leq \frac{4 \ |V| \ |E|^2}{\beta \mathcal{L}(\Gamma^{\delta'})}. 
\end{equation*}
If we let $\alpha'_m=\infty$ at each vertex then, $\lambda_1(\Gamma)=0$. Maybe we can increase the eigenvalue and push it towards the upper bound by adding the maximum number of edges with a large length and minimum vertices. %or by constraining value of function's derivative at vertices.
\end{remark}

Since there is no restriction on the domain of the quadratic form, so we could have used other functions that belongs to $H^1(\Gamma^{\delta'})$ as a test functions i.e. $\cos(x), \sin(x)$, but the bounds that we get using these functions are not so explicitly nice. Let us use the functions $\varphi(x)=\sin(\frac{\pi x}{\ell_j}), \quad x \in e_j=[0,\ell_j]$ and $\varphi(x)=\cos(\frac{\pi x}{\ell_j}), \quad x \in e_j=[\frac{-\ell_j}{2},\frac{\ell_j}{2}]$. Then, 
\begin{equation*}
    \lambda_1(\Gamma^{\delta'}) \leq \frac{\pi^2}{L(\Gamma^{\delta'})} \sum\limits_{j=1}^{|E|}\frac{1}{\ell_j}.
\end{equation*}
If the length of each edge is equal, say, $\ell$, then 
\begin{equation*}
    \lambda_1(\Gamma^{\delta'}) \leq \left(\frac{\pi}{\ell} \right)^2.
\end{equation*}
We can also use the function $\varphi(x)=\sin(\frac{2 \pi x}{\ell_j}), \quad x \in e_i=[0,\ell_j]$, as a test function and we get an upper bound
\begin{equation*}
    \lambda_1(\Gamma^{\delta'} ) \leq \frac{4 \pi^2}{L(\Gamma^{\delta'})} \sum\limits_{j=1}^{|E|} \frac{1}{\ell_j}.
\end{equation*}
Let us consider the function $\varphi(x)=\cos(\frac{2 \pi x}{\ell_j}), \quad x \in e_j=[0,\ell_j]$. Clearly, this function is in the domain of the quadratic form and can be used as a test function to get an upper bound on first eigenvalue and $\varphi(0)=\varphi(\ell_j)=1$ implies that $\sum \limits_{x_i \in v_m}\varphi(x_i)=d_{v_m}$, for each $m=1,2,\cdots,|V|.$

\begin{align*}
    h^{\delta'}[\varphi]&=\int_{\Gamma^{\delta'}}|\varphi'(x)|^2 dx+\sum\limits_{m=1}^{|V|}\frac{1}{\alpha'_m} \left|\sum\limits_{x_i \in v_m} \varphi(x_i)\right|^2,\\
    &= 2\pi^2 \sum\limits_{j=1}^{|E|} \frac{1}{\ell_j}+\sum\limits_{m=1}^{|V|}\frac{d^2_{v_m}}{\alpha'_m}.
\end{align*}
\begin{equation*}
    <\varphi,\varphi>=\int_{\Gamma^{\delta'}}|\varphi(x)|^2 dx=\sum\limits_{j=1}^{|E|} \frac{\ell_j}{2}=\frac{L(\Gamma^{\delta'})}{2}.
\end{equation*}
\begin{equation*}
    \lambda_1(\Gamma^{\delta'}) \leq \frac{2}{L(\Gamma^{\delta'})} \left( 2\pi^2 \sum\limits_{j=1}^{|E|} \frac{1}{\ell_j}+\sum\limits_{m=1}^{|V|}\frac{d^2_{v_m}}{\alpha'_m}  \right).
\end{equation*}
For any fix real number $\lambda$, the eigenvalue counting function $N_{\Gamma}(\lambda)$ is defined as the number of eigenvalues of a graph $\Gamma$ smaller than $\lambda$. Since the quantum graph $\Gamma$ is finite compact, and the operator is self-adjoint with a discrete spectrum bounded from below. Therefore the value of function $N_{\Gamma}(\lambda)$ is finite. 
$$ N_{\Gamma}(\lambda)= \# \{\lambda_i \in \sigma(\Gamma) : \lambda_i \leq \lambda\}. $$
\begin{theorem} \label{ec}
Let $\Gamma^a$ be finite compact graph of total length $\mathcal{L}(\Gamma^a)$ with $|V|$ the number of vertices and $|E|$ the number of edges. Let each vertex of $\Gamma^a$ be equipped with anti-standard vertex conditions. Then 
\begin{equation} \label{ec1}
    \left(\frac{\pi}{\mathcal{L}(\Gamma^a)} \right)^2 (k-|V|)^2 \leq \lambda_k(\Gamma^a) \leq \left(\frac{\pi}{\mathcal{L}(\Gamma^a)} \right)^2 (k+|E|-1)^2.
\end{equation}
\end{theorem}
\begin{proof}
Since the interaction strengths are zero at each vertex, therefore the quadratic form is non-negative, and thus the eigenvalues $\lambda \geq 0$. 
Hence the lower estimate in \eqref{ec1}
is interesting only if $k > |V|$.\\
 Let $\Gamma^d$ be the graph obtained from $\Gamma^a$ by imposing the Dirichlet Conditions at each vertex of $\Gamma^a$, and let $h^a$ and $h^d$ denotes the quadratic form of the graphs $\Gamma^a$ and $\Gamma^d$ with domains $D(h^a)$ and $D(h^d)$, respectively. Since the expression of the quadratic forms $h^a$ and $h^d$ are same; moreover, if a function $\varphi$ satisfy Dirichlet conditions at some vertex $v$, then it also satisfies the anti-standard vertex condition that vertex. Therefore, the domain $D(h^d)$ is a subspace of $D(h^a)$ and the quadratic forms $h^a$ and $h^d$ agree on $D(h^d)$, thus minimizing over a smaller domain results in large eigenvalues. Furthermore, the domain $D(h^d)$ is a co-dimension one subspace of $D(h^a)$. Thus, by the rank one nature perturbation, the following interlacing inequalities holds.
 $$\lambda_k(\Gamma^a) \leq \lambda_k(\Gamma^d) \leq \lambda_{k+|V|}(\Gamma^a)$$
 Let  for some $\lambda \in \mathbb{R}$, $N_{\Gamma^a}(\lambda)$ and $N_{\Gamma^d}(\lambda)$ denotes the eigenvalue counting functions for the graphs $\Gamma^a$ and $\Gamma^d$, respectively. Then the above interlacing inequalities implies the following relation between eigenvalue counting functions.
 \begin{equation} \label{cfrelation}
     N_{\Gamma^d}(\lambda) \leq N_{\Gamma^a}(\lambda) \leq N_{\Gamma^d}(\lambda) +|V|.
 \end{equation}
Consider the Laplace operator acting on a interval of length $\ell$ , and assume that the end points are equipped with Dirichlet Conditions. The eigenvalues are $\lambda_k=\left(  \frac{k \pi}{\ell}\right)^2$, for $\lambda \geq 0$ the value of eigenvalue counting  function is $N_{[0,\ell]}(\lambda)=\left[\frac{\sqrt{\lambda }}{\pi} \ell \right].$ Where square brackets mean to take the integer part of the argument. Since Dirichlet conditions imposed at a vertex of degree two or more does not connect the individual function living on the incident edge in any way.  The Dirichlet conditions have the effect of disconnecting the vertex of degree $d_v$ into $d_v$ vertices of degree one, so the graph $\Gamma^d$ is now decoupled into a set of intervals, and the set of eigenvalues is just the union of eigenvalues of each interval (counting multiplicities). Moreover, 

The counting function is given by 
\begin{equation} \label{c1}
    N_{\Gamma^{d}}(\lambda)= \sum\limits_{j=1}^{|E|}N_{[0,\ell_j]}(\lambda) =\left[\frac{\sqrt {\lambda }}{\pi} \ell_1 \right] +\left[\frac{\sqrt {\lambda }}{\pi} \ell_2 \right] + \cdots+ \left[\frac{\sqrt {\lambda }}{\pi} \ell_{|E|} \right] \leq \left[\frac{\sqrt {\lambda }}{\pi} \mathcal{L}(\Gamma^a) \right].
\end{equation}
 Since for any $a$ and $b$, $ [a]+[b] \leq [a+b]$, therefore taking integer part of the sum of the terms is increased at most by the number of terms minus one as compared to adding integer parts only. As the number of terms in \eqref{c1} are equal to number of edges $|E|$, thus 
 $$\left[\frac{\sqrt {\lambda }}{\pi} \mathcal{L}(\Gamma^a) \right] =\left[\frac{\sqrt {\lambda }}{\pi} (\ell_1 +\ell_2+\cdots+\ell_{|E|})\right] \leq   \left[\frac{\sqrt {\lambda }}{\pi} \ell_1 \right] +\left[\frac{\sqrt {\lambda }}{\pi} \ell_2 \right] + \cdots+ \left[\frac{\sqrt {\lambda }}{\pi} \ell_{|E|} \right] +|E|-1, $$
 and we have 
 \begin{equation} \label{c2}
     \left[\frac{\sqrt {\lambda }}{\pi} \mathcal{L}(\Gamma^a) \right] -|E|+1 \leq N_{\Gamma^{d}}(\lambda).
 \end{equation}
 \end{proof}
 Thus, from \eqref{cfrelation},\eqref{c1} and \eqref{c2}, we obtain the following inequalities.
 \begin{equation*}
     \left[\frac{\sqrt {\lambda }}{\pi} \mathcal{L}(\Gamma^a) \right] -|E|+1 \leq N_{\Gamma^{d}}(\lambda) \leq N_{\Gamma^{a}}(\lambda) \leq N_{\Gamma^d}(\lambda) +|V| \leq \left[\frac{\sqrt {\lambda }}{\pi} \mathcal{L}(\Gamma^a) \right]+|V|.
 \end{equation*}
 Setting $\lambda=\left (\dfrac{k \pi } {\mathcal{L}(\Gamma^a)} \right)^2$ we get 
 $$k-|E|+1 \leq N_{\Gamma^{a}} \left (\frac{k \pi } {\mathcal{L}(\Gamma^a)} \right)^2 \leq k+|V|,$$ so
 $$\lambda_{k - |E|+1} \leq \left ( \frac{k \pi } {\mathcal{L}(\Gamma^a)} \right)^2 \leq \lambda_{k+|V|}.$$
 This estimate implies that the multiplicity of the eigenvalues is uniformly bounded by $|E|+|V|$.
 setting $\tilde{k}=k-|E|+1$ we get $\lambda_{\tilde{k}} \leq \left(\dfrac{\pi}{\mathcal{L}(\Gamma^a)}\right)^2 (\tilde{k}+|E|-1)^2$ and similarly setting $\tilde{k}=k+|V|$ we get $\left(\dfrac{\pi}{\mathcal{L}(\Gamma^a)}\right)^2 (\tilde{k}-|V|)^2 \leq \lambda_{\tilde{k}}$. Since $\tilde{k}$ is a dummy subscript so we can replace it with $k$ and finally, we obtain,
 $$\left(\frac{\pi}{\mathcal{L}(\Gamma^a)} \right)^2 (k-|V|)^2 \leq \lambda_k(\Gamma^a) \leq \left(\frac{\pi}{\mathcal{L}(\Gamma^a)} \right)^2 (k+|E|-1)^2.$$

 Another type of vertex condition that decouples the graph into disjoints interval is the Neumann conditions. For these conditions, the derivatives of functions living on incident edges are zero at each vertex, and no conditions are assumed on the values of functions. That is, $$\varphi_e'(v)=0$$
 Let $\Gamma^n$ be the finite compact metric graph with Neumann vertex conditions, then the quadratic form is defined by, 
 $$h^n(\varphi)=\int_{\Gamma^{n}} |\varphi'|^2dx$$
 with domain $D(h^n)=H^1(\Gamma^{n} \backslash V)$. The domain $D(h^a)$ is a co-dimension one subspace of $D(h^n)$, and the quadratic forms $h^a$ and $h^n$ agree on $D(h^a)$. Thus, by rank one perturbation, the following interlacing inequalities holds.
 $$\lambda_k(\Gamma^n) \leq \lambda_k(\Gamma^a) \leq \lambda_{k+|V|}(\Gamma^n)$$
 Let  for some $\lambda \in \mathbb{R}$, $N_{\Gamma^{a}}(\lambda)$ and $N_{\Gamma^{n}}(\lambda)$ denotes the eigenvalue counting functions for the graphs $\Gamma^a$ and $\Gamma^n$, respectively. Then the above interlacing inequalities implies the following relation between eigenvalue counting functions.
 \begin{equation} 
     N_{\Gamma^a}(\lambda) \leq N_{\Gamma^n}(\lambda) \leq N_{\Gamma^a}(\lambda) +|V|.
 \end{equation}
 In the previous theorem, if we impose Neumann conditions instead of Dirichlet conditions to decouple the graph $\Gamma^a$, then one can prove the following theorem correspondingly.
 \begin{theorem} \label{anti2}
Let $\Gamma^a$ be finite compact graph of total length $\mathcal{L}(\Gamma^a)$ with $|V|$ the number of vertices and $|E|$ the number of edges. Let each vertex of $\Gamma^a$ be equipped with anti-standard vertex conditions. Then 
\begin{equation} 
    \left(\frac{k \pi}{\mathcal{L}(\Gamma^a)} \right)^2  \leq \lambda_k(\Gamma^a) \leq \left(\frac{\pi}{\mathcal{L}(\Gamma^a)} \right)^2 (k+|E|+|V|-1)^2.
\end{equation}
\end{theorem}
\begin{theorem} \label{deltap}
Let $\Gamma^{\delta'}$ be metric graph with negative strengths $\alpha'_m <0$ at each vertex. Then,
\begin{equation}
    \left(\frac{(k-2|V|)\pi}{2 \mathcal{L}(\Gamma^{\delta'})}\right)^2 \leq \lambda_k(\Gamma^{\delta'}), \quad k \geq 2|V|.
\end{equation}
\end{theorem}
\begin{proof}
Let $\hat{\Gamma}^{\delta'}$ be a double cover obtained by doubling each edge in $\Gamma^{\delta'}$ and assigning strengths $2\alpha'_m$ at each vertex. Then each eigenvalue of $\Gamma^{\delta'}$ is also an eigenvalue of $\hat{\Gamma}^{\delta'}$. However, the converse is not true. Let $\lambda(\Gamma^{\delta'})$ be an eigenvalue corresponding to the eigenfunction $\varphi$. The eigenfunction $\varphi$ can be lifted on $\hat{\Gamma}^{\delta'}$ by letting it assume the same values on the new edges as on edges of $\Gamma^{\delta'}$. Then newly constructed function $\hat{\varphi}$ is an eigenfunction for $\hat{\Gamma}^{\delta'}$ corresponding to eigenvalue $\lambda(\Gamma^{\delta'})$, because for $\lambda(\Gamma^{\delta'})$ this function satisfies the eigenvalue equation and vertex conditions on $\hat{\Gamma}^{\delta'}$. Thus, 
$$\lambda_k(\hat{\Gamma}^{\delta'}) \leq \lambda_k(\Gamma^{\delta'}).$$
Let $\mathcal{C}^{\delta'}$ be the graph obtained from $\hat{\Gamma}^{\delta'}$ by splitting each of the vertices of $\hat{\Gamma}^{\delta'}$ to form a cycle of length $2 \mathcal{L}(\mathcal{C}^{\delta'})$, since each vertex has strength $2\alpha'_m<0 $, therefore [\textcolor{blue}{Theorem} \eqref{roh}] implies,
\begin{equation*}
\lambda_k(\mathcal{C}^{\delta'}) \leq \lambda_k(\hat{\Gamma}^{\delta'}).    
\end{equation*}
Obtain a graph $\mathcal{C}^{a}$ by assigning the zero at all vertices of $\mathcal{C}^{\delta'}$. Then by \textcolor{blue}{Theorem} \eqref{strength}, 
$$\lambda_{k-2|V|}(\mathcal{C}^a)\leq \lambda_k(\mathcal{C}^{\delta'}).$$
The graph $\mathcal{C}^a$ is a bipartite metric graph, therefore \cite[\textcolor{blue}{Theorem 3.6}]{PR} implies that 
$$\lambda_{k-2|V|}(\mathcal{C}^s)=\lambda_{k-2|V|}(\mathcal{C}^a)$$
The graph $\mathcal{C}$ is a loop of length $2\mathcal{L}(\Gamma^{\delta'})$, and now
we cut the graph $\mathcal{C}^s$ to form an edge $I$ of length $2\mathcal{L}(\Gamma^{\delta'}).$ Thus,
$$\left(\frac{(k-2|V|)\pi}{2 \mathcal{L}(\Gamma^{\delta'})}\right)^2=\lambda_{k-2|V|}(I)\leq \lambda_{K-2|V|}(\mathcal{C}^s) .$$

\end{proof}
A similar lower bound on the eigenvalues of the Laplacian acting on the edges of a star graph with anti-standard vertex condition was proved in \cite[\textcolor{blue}{Theorem 1} ]{ZS}. We will use it to get a lower estimate on the metric graph $\Gamma$ with $\delta'$-conditions.
\begin{theorem} \label{deltaprime}
Let $\Gamma^{\delta'}$ be finite compact connected metric graph with $|E| \geq 2$, and $\delta'$-conditions equipped at each vertex with total strength $\alpha'=\sum\limits^{|V|}_{m=1} \alpha'_m$. Let $\ell_1 \geq \ell_2 \geq \cdots \geq\ell_{|E|}$ denotes the length of each edge $e_j$, then for $k,|E| \in \mathbb{N}$, and $1 \leq j \leq |E|$, we have
    %\item  If $\alpha'_m \geq 0$ for all $m=1,2,\cdots,|V|$, then 

%\begin{equation}
 %       min \left \{  \frac{\left(2\left( k- \left( 1+\frac{1}{|E|}\right) \right)+3\right)^2 \pi^2}{4 \ell_{1}^2}, \frac{\left(2\left( k- \left( 1+\frac{1}{|E|}\right) \right)+1\right)^2 \pi^2}{4 \ell_{j}^2} \right \}    \leq    \lambda_{k|E|+j}(\Gamma^{\delta'}) 
%\end{equation}
%\item If $\alpha' >0 $, then 
%\begin{equation}
%        min \left \{\frac{\left(2\left( k- \left( 1+\frac{|V|}{|E|}\right) \right)+3\right)^2 \pi^2}{4 \ell_1^2}, \frac{\left(2\left( k- \left( 1+\frac{|V|}{|E|}\right) \right)+1\right)^2 \pi^2}{4 \ell_j^2} \right \}    \leq    \lambda_{k|E|+j}(\Gamma^{\delta'}) 
%\end{equation}
%\item 
\begin{equation}
        min \left \{\frac{\left(2\left( k- \left( 1+\frac{|V|+1}{|E|}\right) \right)+3\right)^2 \pi^2}{4 \ell_1^2}, \frac{\left(2\left( k- \left( 1+\frac{|V|+1}{|E|}\right) \right)+1\right)^2 \pi^2}{4 \ell_j^2} \right \}    \leq    \lambda_{k|E|+j}(\Gamma^{\delta'}) 
\end{equation}
\end{theorem}
 \begin{proof}
 \begin{enumerate}
 Let $S^a$ denote the star graph constructed  from the edges of $\Gamma^{\delta'
 }$ and let each vertex of $S^a$ be equipped with anti-standard vertex condition. Then by \cite[\textcolor{blue}{Theorem 1} ]{ZS}
 \begin{equation}
        min \left \{  \frac{(2k+3)^2 \pi^2}{4 \ell_1^2}, \frac{(2k+1)^2 \pi^2}{4 \ell_j^2} \right \}    \leq    \lambda_{k|E|+j}(S^a). 
\end{equation}
Construct a new star graph $S^{\delta'}$ by assigning the negative strengths at vertices of $S^a$ . Then by \textcolor{blue}{Theorem} \eqref{strength}, 
$$\lambda_{k|E|+j}(S^a) \leq \lambda_{k|E|+j+|E|+1}(S^{\delta'}).$$
Now, first create a flower graph $\Gamma^{\delta'}_f$ from $S^{\delta'}$ in accordance to \textcolor{blue}{Theorem}\eqref{roh}(2). Thus, 
$$\lambda_{k|E|+j+|E|+1}(S^{\delta'}) \leq \lambda_{k|E|+j+|E|+1}(\Gamma^{\delta'}_f).$$
and then obtain a graph $\Gamma^a_f$ by prescribing anti-standard conditions to $\Gamma^{\delta'}_f$.\\  
We restore the vertices of $\Gamma^{\delta'}$ by cutting, step by step, the vertex of flower graph $\Gamma^a_f$, then by  \textcolor{blue}{Theorem}\eqref{roh}(6), the obtained graph is same as $\Gamma^{\delta'}$ but equipped with anti-standard conditions. Call this new graph $\Gamma^a$, Thus
$$\lambda_{k|E|+j+|E|+1}(\Gamma^{\delta'}_f) \leq \lambda_{k|E|+j+|E|+1}(\Gamma^a_f)\leq \lambda_{k|E|+j+|E|+1}(\Gamma^a).$$
Now, assign the strengths at all vertices of $\Gamma^a$ to obtain the graph $\Gamma^{\delta'}$. Thus $$\lambda_{k|E|+j+|E|+1}(\Gamma^a) \leq \lambda_{k|E|+j+|E|+|V|+1}(\Gamma^{\delta'}).$$  
\begin{align*}
k|E|+j+|E|+|V|+1&=k|E|+|E|+|V|+1+j \\
&=k|E|+\frac{(|E|+|V|+1)|E|}{|E|}+j \\
&= \left(k+ \frac{(|E|+|V|+1)}{|E|}\right)|E|+j 
\end{align*}
Let $k'=\left(k+ \frac{(|E|+|V|+1)}{|E|}\right)$ $\implies$ $k=k'-\left(1+\frac{|V|+1}{|E|} \right).$
\begin{equation}
        min \left \{  \frac{(2k+3)^2 \pi^2}{4 \ell_1^2}, \frac{(2k+1)^2 \pi^2}{4 \ell_j^2} \right \}    \leq    \lambda_{k|E|+j+|E|+|V|+1}(\Gamma^{\delta'}). 
\end{equation}
\begin{equation*}
        min \left \{\frac{\left(2\left( k'- \left( 1+\frac{|V|+1}{|E|}\right) \right)+3\right)^2 \pi^2}{4 \ell_1^2}, \frac{\left(2\left( k'- \left( 1+\frac{|V|+1}{|E|}\right) \right)+1\right)^2 \pi^2}{4 \ell_j^2} \right \}    \leq    \lambda_{k'|E|+j}(\Gamma^{\delta'}) 
\end{equation*}
 \end{enumerate}
 \end{proof}
 
 \begin{definition}
A {\bf spanning tree }$T$ of a graph $\Gamma$ is an undirected subgraph of $\Gamma$ containing maximal edges of $\Gamma$ with no cycles, or it contains all the vertices of $\Gamma$  connected by a minimal set of edges of $\Gamma$.
\end{definition}
A graph $\Gamma$ can have more than one spanning tree, each having an equal number of edges, let $|E_T|$ denotes the cardinality of the edge set of spanning tree $T$. Each spanning tree has it's own length so we can compare their lengths, we define a new term for our later use.
\begin{definition}
A {\bf maximal spanning tree} $T_m$ is a spanning tree for a graph $\Gamma$ having the largest length among all the spanning trees of $\Gamma$.
\end{definition}
For a finite, compact connected metric tree $T^s$ with standard conditions, let $P^s$ be the largest path inside $T^s$.  Since the spectrum of Laplacian on the path graph $P^s$ with standard conditions consists of discrete points $$\lambda_{k+1}(P^s)=\left(\frac{k \pi}{\mathcal{L}(P^s)}\right)^2, \quad k=0,1,2,\cdots$$
The tree graph $T^s$ can be obtained by attaching pendant edges with standard conditions to the vertices of $P^s$ in accordance to \cite[\textcolor{blue}{Theorem 3.10}]{BKKM}. Thus, 
\begin{equation}
\lambda_{k+1}(T^s) \leq \lambda_{k+1}(P^s)=\left(\frac{k \pi}{\mathcal{L}(P^s)}\right)^2, \quad k=0,1,2,\cdots    
\end{equation}

Now consider a finite compact connected metric tree $T^a$ with anti-standard conditions and
let $P^a$ be the longest path inside $T^a$.  Since the spectrum of Laplacian on the path graph $P^a$ with anti-standard conditions consists of discrete points $$\lambda_{k}(P^a)=\left(\frac{k \pi}{\mathcal{L}(P^a)}\right)^2, \quad k=1,2,\cdots$$
The tree graph $T^a$ can be obtained by attaching pendant edges with anti-standard conditions to the vertices of $P^a$ in accordance to \cite[\textcolor{blue}{Theorem 3.5}]{RS}. Thus $$\lambda_{k}(T^a) \leq \lambda_{k}(P^a)=\left(\frac{k \pi}{\mathcal{L}(P^a)}\right)^2, \quad k=1,2,\cdots$$ 
Or using  \textcolor{blue}{Theorem}\eqref{pendant}
$$\lambda_{k+1}(T^a) \leq \lambda_{k}(P^a)=\left(\frac{k \pi}{\mathcal{L}(P^a)}\right)^2, \quad k=1,2,\cdots$$ 
 One can obtain these estimates with any path graph inside the metric tree, even an interval would give an estimate. However, we choose the longest path to get better estimates. The weak version of these bounds were established in \cite[\textcolor{blue}{Theorem 3.2}]{R} for standard conditions and in \cite[\textcolor{blue}{Theorem 3}]{ZS} for anti-standard conditions, in which a particular path graph  $P$ was chosen which starts with the longest edge and ends at second largest edge. These estimates were further expressed in terms of average length of metric tree.
 For standard conditions,
 \begin{equation} \label{tree1}
     \lambda_{k+1}(T^s) \leq \lambda_{k+1}(P^s)=\left(\frac{k \pi}{\mathcal{L}(P^s)}\right)^2 \leq \left(\frac{k \pi |E|}{2 \mathcal{L}(T^s)}\right)^2, \quad k=0,1,2,\cdots
 \end{equation}
Similarly, for anti-standard conditions,
 \begin{equation} \label{tree2}
     \lambda_{k}(T^a) \leq \lambda_{k}(P^a)=\left(\frac{k \pi}{\mathcal{L}(P^a)}\right)^2 \leq \left(\frac{k \pi |E|}{2 \mathcal{L}(T^a)}\right)^2, \quad k=1,2,\cdots
 \end{equation}
The following simple but useful relation was proved in \cite[\textcolor{blue}{Corollary 3.8} ]{PR}, which connects the eigenvalues of metric tree with standard and anti-standard vertex conditions. 
\begin{equation}
    \lambda_k(T^a)=\lambda_{k+1}(T^s), \quad k\ \in \mathbb{N}.
\end{equation}
We use these estimates to obtain some bounds on the spectrum of the Laplacian acting on a finite compact connected metric graph $\Gamma$ with standard and anti-standard vertex conditions, and then we will generalize these estimates for $\delta$ and $\delta'$- type vertex conditions. Most of the following estimates are equal in magnitude but are expressed in terms of different geometrical parameters of the metric graph $\Gamma$.
\\ For the anti-standard vertex conditions, we know that either attaching a pendant graph (particularly an edge), or attaching an edge between existing vertices lowers all eigenvalues \cite[\textcolor{blue}{Theorem 3.2 , 3.5} ]{RS}. Therefore,
\begin{equation*}
    \lambda_{k}(\Gamma^a) \leq \lambda_k(T^a) \leq \lambda_k(P^a)=\left(\frac{k \pi \gamma}{\mathcal{L}(\Gamma^a)}\right)^2, \quad k=1,2,\cdots
\end{equation*}
Where $\gamma=\frac{\mathcal{L}(\Gamma^a)}{\mathcal{L}(P^a)}$.
If we let $\gamma=\mathcal{L}(\Gamma^a)-\mathcal{L}(P^a)$, then the same estimate obtained by the path graph can be rewritten as 
\begin{equation*}
    \lambda_{k}(\Gamma^a) \leq  \lambda_k(P^a)=\left(\frac{k \pi }{\mathcal{L}(\Gamma^a)-\gamma}\right)^2, \quad k=1,2,\cdots
\end{equation*}
Note that the above bounds are equal to $\lambda_{k}(P^a)=\left(\frac{k \pi}{\mathcal{L}(P^a)}\right)^2$ in magnitude, but are expressed in terms of total length of the graph.
 \\ Let $T^a$ be the tree graph inside $\Gamma^a$ with the longest length among all the tree inside $\Gamma^a$. Then we use the estimate \eqref{tree2} on $T^a$ to obtain an estimate on $\Gamma^a$.
 $$\lambda_k(\Gamma^a) \leq \lambda_k(T^a) \leq \left(\frac{k \pi |E_T|}{2 \mathcal{L}(T^a)}\right)^2, \quad k=1,2,\cdots$$
 Moreover, if the graph $\Gamma^a$ with $|E|$ number of edges is obtained by from metric tree $T^a$ by  either attaching edges or loops of lengths $\ell_1, \ell_2, \cdots, \ell_m$, then 
 $$\lambda_k(\Gamma^a) \leq \left(\frac{k \pi (|E|-m)}{2 \left(\mathcal{L}(\Gamma^a)-\sum\limits_{i=1}^m \ell_i\right)}\right)^2, \quad k=1,2,\cdots$$
 If $T^a_m$ is the maximal spanning tree inside $\Gamma^a$, then 
 $$\lambda_k(\Gamma^a) \leq \lambda_k(T^a_m) \leq \left(\frac{k \pi |E_{T_m}|}{2 \mathcal{L}(T^a_m)}\right)^2, \quad k=1,2,\cdots$$
 The number of vertices in $T^a_m$ are equal to the number of vertices in $\Gamma^a$, and the number of edges in $T^a_m$ are one less than the number of vertices in $\Gamma^a$. Let $|V|$ denotes the number of vertices in $\Gamma^a$, and the graph $\Gamma^a$ can be obtained from $T^a_m$ by attaching pendant loops or multiple edges between existing vertices. Thus 
 \begin{equation}
     \lambda_k(\Gamma^a) \leq \lambda_k(T^a_m) \leq \left(\frac{k \pi \gamma (|V|-1)}{2 \mathcal{L}(\Gamma^a)}\right)^2, \quad k=1,2,\cdots
 \end{equation}
 The above bounds can be expressed in terms of the first Betti number of the graph $\Gamma^a$. Since the first Betti number of maximal spanning tree is zero, and while constructing $\Gamma^a$ from $T^a_m$ each attachment of either loop or an edge between existing vertices increases the Betti number by one. Let $\beta$ be number of cycles in graph $\Gamma^a$ and let $\ell_1, \ell_2, \cdots, \ell_\beta$ be the number of required edges to obtain $\Gamma^a$ from $T^a_m$. Then, 
 $$\lambda_k(\Gamma^a) \leq \left(\frac{k \pi (|E|-\beta)}{2 \left(\mathcal{L}(\Gamma^a)-\sum\limits_{i=1}^m \ell_i\right)}\right)^2, \quad k=1,2,\cdots$$
 In addition, if $\ell_1= \ell_2= \cdots= \ell_\beta$, then 
 $$\lambda_k(\Gamma^a) \leq \left(\frac{k \pi (|E|-\beta)}{2 \left(\mathcal{L}(\Gamma^a)- \beta \ell \right)}\right)^2, \quad k=1,2,\cdots$$
 For different types of metric tree inside the graph $\Gamma^a$, we obtain different expression for bounds, some of these were equal in magnitude but expressed in different topological parameters. Now we chose a particular type of a metric tree inside $\Gamma^a$, and will use the estimate \eqref{tree2}.  Let $T^a$ be the tree inside $\Gamma^a$ such that the number of edges and the length of each edge in $T^a$ and $\Gamma^a$ are equal. Then $\Gamma^a$ can be obtained by gluing some vertices of $T^a$ and each gluing produces a cycle. Thus 
 \begin{equation} \label{anti1}
     \lambda_k(\Gamma^a) \leq \lambda_k(T^a) \leq  \left(\frac{k \pi |E|}{2 \mathcal{L}(\Gamma^a)}\right)^2, \quad k=1,2,\cdots
 \end{equation}
 Similarly, for the graph $\Gamma^s$ with standard vertex condition the following bounds hold.
 Let $P^s$ be the longest path inside $\Gamma^s$, then we attach the pendant edges, step by step, with standard condition to the vertices of $P^s$ unless we obtain a tree $T^s$ with number of edges equal to the number of edges $|E|$ of graph $\Gamma^s$. Now $\Gamma^s$ can be obtained by pairwise gluing of $\beta$ pair of vertices.
 \begin{equation*}
     \lambda_{k+1-\beta}(\Gamma^s) \leq \lambda_{k+1} (T^s) \leq \lambda_{k+1} (P^s)=\left(\frac{k \pi }{\mathcal{L}(P^s)}\right)^2, \quad k=1,2,\cdots,
 \end{equation*}
 $$\lambda_{k+1}(\Gamma^s) \leq \left(\frac{(k+\beta) \pi \gamma}{\mathcal{L}(\Gamma^s)}\right)^2, \quad k=1,2,\cdots$$
 where $\gamma=\frac{\mathcal{L}(\Gamma^s)}{L(P^s)}$. 
 If we consider the longest metric tree $T^s$, then we can use the estimate \eqref{tree1}, 
 $$\lambda_{k+1}(\Gamma^s) \leq \lambda_{k+1}(T^s) \leq \left(\frac{(k+\beta) \pi |E_T|}{2\mathcal{L}(T^s)}\right)^2=\left(\frac{(k+\beta) \pi |E_T|\gamma}{2\mathcal{L}(\Gamma^s)}\right)^2, \quad k=1,2,\cdots$$
 If we choose a tree $T^s$ such that each vertex of $T^s$ is equipped with standard vertex conditions and the number of edges and the length of each edge is same as of $\Gamma^s$. However, the only difference between $T^s$ and $\Gamma^s$ is that $\Gamma^s$ can be obtained from $T^s$ by pairwise gluing of $\beta$ pair of vertices, then 
 \begin{equation} \label{standard2}
     \lambda_{k+1}(\Gamma^s) \leq \left(\frac{(k+\beta) \pi |E|}{2 \mathcal{L}(\Gamma^s)}\right)^2
 \end{equation}
 If $T^s_m$ is the maximal spanning tree inside $\Gamma^s$, then
 $$\lambda_{k+1}(\Gamma^s) \leq \left(\frac{(k+\beta) (|V|-1)\pi }{2 \left(\mathcal{L}(\Gamma^s)-\sum\limits_1^{m} \ell_i\right)}\right)^2.$$
 Or 
 $$\lambda_{k+1}(\Gamma^s) \leq \left(\frac{(k+\beta) (|E|-\beta)\pi }{2 \left(\mathcal{L}(\Gamma^s)-\sum\limits_1^{m} \ell_i\right)}\right)^2.$$
 Where $|E|$, $|V|$ and $\beta$ denotes the number of edges, number of vertices and the first Betti number of the graph $\Gamma^s$.\\
 \begin{definition}
 Let $\Gamma$ be a finite compact connected metric graph. Then diameter of $\Gamma$ is defined as, 
 $$d(\Gamma):=\max \{ d(x,y) : x,y \in \Gamma \}.$$
 \end{definition}
 Let $T^s$ be a finite, connected metric tree and let $d(T^s)$ denote the diameter of $T^s$, then the following inequality was proved in \cite[\textcolor{blue}{Theorem 3.4}]{R}
 \begin{equation} \label{diameter1}
     \lambda_{k+1}(T^s) \leq  \left(\frac{k \pi}{d(T^s)}\right)^2.
 \end{equation}
 A similar inequality for $T^a$ was proved in  \cite[\textcolor{blue}{Theorem 3}]{ZS},
  \begin{equation} \label{diameter2}
     \lambda_{k}(T^a) \leq  \left(\frac{k \pi}{d(T^a)}\right)^2.
 \end{equation}
 In order to establish better estimates on the eigenvalues of metric trees $T^s$ and $T^a$, we require the following lemma which is also contained in \cite[\textcolor{blue}{Lemma 4.2}]{RL}. However, we provide a brief proof for the completeness of the presentation.
 \begin{lemma}
 Let $T$ be a metric tree of total length $\mathcal{L}(T)$ with $|E_p|\geq 2$ the pendant edges, then 
 \begin{equation}
            \frac{2 \mathcal{L}(T)}{|E_p|} \leq  d(T). 
 \end{equation}
 \end{lemma}
 \begin{proof}
 We will show that the diameter of a metric tree is the distance between two pendant points.
 Suppose that $d(\Gamma)=d(x_1,x_2)$ for some $x_1,x_2 \in T$. To prove by contradiction we assume that $x_1$ is not a pendant point of $T$, then $x_1$ is some interior point and the graph $T \backslash x_1$ has at least two connected components. Let $T_1$ denote of the component such that $x_2 \notin T_1$ and $x_1\neq y \in T_1$. Then $$d(y,x_2) > d(x_1,x_2)=d(T).$$
 This holds because $T$ is a tree and any path inside $T$ that between $x_2$ and $y$ contains the point $x_1$. This contradicts our assumption of $d(x_1,x_2)=d(T)$. 
 Let $x_0$ be a middle point of a shortest path $P$ between $x_1$ and $x_2$. Thus, 
 $$d(x_1,x_0)=\frac{d(T)}{2}=d(x_0,x_2).$$
 Since the tree $T$ has $E_p$ pendant edges, therefore there are exactly $E_p$ path of length at most $d(x_1,x_0)$, that connect the point $x_0$ and a pendant vertex of $T$. Moreover, these paths covers the metric tree $T$, thus
 $$\mathcal{L}(T) \leq \sum\limits_{v \text{ is a pendant}} d(x_0,v) \leq \sum\limits_{v \text{ is a pendant}} d(x_0,x_1)= \frac{E_p d(T)}{2}.$$
 \end{proof}
 The bounds \eqref{diameter1} and \eqref{diameter2} implies that  
 \begin{equation} \label{newtree1}
     \lambda_{k+1}(T^s) \leq \left( \frac{k \pi |E_p|}{2 \mathcal{L}(T^s)} \right)^2 
 \end{equation}
 and,
 \begin{equation} \label{newtree2}
     \lambda_{k}(T^a) \leq \left( \frac{k \pi |E_p|}{2 \mathcal{L}(T^a)} \right)^2.
     \end{equation}
 Thus, we get improved estimates as compared to the estimates \eqref{tree1} and \eqref{tree2} on the eigenvalues of a metric trees $T^s$ and $T^a$. Therefore, the above listed bounds can be improved when the estimates \eqref{tree1} and \eqref{tree2} are replaced by \eqref{newtree1} and \eqref{newtree2}. \\
 We now use previously established estimates on metric graph with standard conditions to obtain some estimates on the metric graph $\Gamma^{\delta}$. For two same underlying metric graph $\Gamma^{\delta}_\alpha$ and $\Gamma^{\delta}_{\alpha'}$ with strengths $\alpha$ and $\alpha'$ at some vertex $v_0$, respectively, the following interlacing inequality holds. 
 $$\lambda_k(\Gamma^{\delta}_\alpha) \leq \lambda_k(\Gamma^{\delta}_{\alpha'}) \leq \lambda_{k+1}(\Gamma^{\delta}_\alpha).$$
 Let $\Gamma^{\delta}$ be finite compact connected metric graph, and let $|V^p|$ and $|V^n|$ denotes the number of vertices of $\Gamma^{\delta}$ with positive and negative strengths, respectively, and let $\Gamma^s$ be the same underlying metric graph with standard conditions. If $\Gamma^p$ is a graph obtained from $\Gamma^s$ by replacing the interaction strengths from zero to positive at those vertices where positive strengths were specified in $\Gamma^{\delta}$. Now $\Gamma^{\delta}$ can be obtained from this new graph $\Gamma^p$ by assigning negative strengths at some vertices, where negative strengths were assigned in $\Gamma^{\delta}$.  Then following inequalities holds.
 $$\lambda_k(\Gamma^s) \leq \lambda_k(\Gamma^p) \leq \lambda_{k+|V^p|}(\Gamma^s),$$
 $$\lambda_k(\Gamma^{\delta}) \leq \lambda_k(\Gamma^p) \leq \lambda_{k+|V^n|}(\Gamma^{\delta})$$
 $$\lambda_k(\Gamma^{\delta}) \leq \lambda_k(\Gamma^p) \leq \lambda_{k+|V^p|}(\Gamma^s) \leq \left(\frac{(k+\beta+|V^p|-1) \pi |E|}{2 \mathcal{L}(\Gamma^{\delta})}\right)^2 $$
 Where last inequality is due to \eqref{standard2}. If we replace the zero strengths with negative strengths first and then with the positive strengths even then we obtain the same bound.
 $$ \lambda_k(\Gamma^n) \leq \lambda_k(\Gamma^s)  \leq \lambda_{k+|V^n|}(\Gamma^n)$$
 $$\lambda_k(\Gamma^n) \leq \lambda_k(\Gamma^{\delta})  \leq \lambda_{k+|V^p|}(\Gamma^n) $$
 \begin{equation}
     \lambda_k(\Gamma^{\delta}) \leq \lambda_{k+|V^p|}(\Gamma^n) \leq \lambda_{k+|V^p|}(\Gamma^s) \leq \left(\frac{(k+\beta+|V^p|-1) \pi |E|}{2 \mathcal{L}(\Gamma^{\delta})}\right)^2
 \end{equation}
 The other way to establish an upper estimate for eigenvalues of the metric graph $\Gamma^{\delta}$ is using the application of the corresponding flower graph. Let $\Gamma^{\delta}$ be the metric graph with total interaction strength $\alpha=\sum\limits_{m=1}^{|V|} \alpha_m$, and $\Gamma^s$ be same graph with standard conditions at each vertex. Let $\Gamma^s_f$ and $\Gamma^{\alpha}_f$ be the corresponding flower graph with strengths zero and $\alpha$.\\
 If $\alpha >0$, then 
 $$\lambda_k(\Gamma^s) \leq \lambda_k(\Gamma^s_f) \leq \lambda_{k+|V|-1}(\Gamma^s)$$
  $$\lambda_k(\Gamma^s_f) \leq \lambda_k(\Gamma^\alpha_f) \leq
  \lambda_{k+1}(\Gamma^s_f)$$
 Now restore the vertices of $\Gamma^{\delta}$ from $\Gamma^\alpha_f$, then 
 $$\lambda_k(\Gamma^{\delta}) \leq \lambda_k(\Gamma^\alpha_f) \leq \lambda_{k+|V|-1}(\Gamma^{\delta})$$

  $$\lambda_k(\Gamma^{\delta}) \leq \lambda_k(\Gamma^\alpha_f) \leq \lambda_{k+1}(\Gamma^s_f) \leq \lambda_{k+|V|}(\Gamma^s) \leq \left(\frac{(k-1+\beta+|V|) |E| \pi}{2 L(\Gamma^{\delta})}\right)^2. $$
  If $\alpha \leq0$,
 $$\lambda_k(\Gamma^s) \leq \lambda_k(\Gamma^s_f) \leq \lambda_{k+|V|-1}(\Gamma^s)$$
  $$\lambda_k(\Gamma^\alpha_f) \leq \lambda_k(\Gamma^s_f) \leq
  \lambda_{k+1}(\Gamma^\alpha_f)$$
 Now restore the vertices of $\Gamma^{\delta}$ from $\Gamma^\alpha_f$, then 
 $$\lambda_k(\Gamma^{\delta}) \leq \lambda_k(\Gamma^\alpha_f) \leq \lambda_{k+|V|-1}(\Gamma^{\delta})$$
 
  $$\lambda_k(\Gamma^{\delta}) \leq \lambda_k(\Gamma^\alpha_f) \leq \lambda_{k}(\Gamma^s_f) \leq \lambda_{k+|V|-1}(\Gamma^s) \leq \left(\frac{(k-2+\beta+|V|) |E| \pi}{2 L(\Gamma^{\delta})}\right)^2. $$ 
 Let $\Gamma_a^{\delta'}$ be the finite compact connected metric graph equipped with $\delta'$-type conditions and possibly anti-standard conditions at some vertices. Let $\Gamma^a$ be the same underlying metric graph as $\Gamma_a^{\delta'}$ with anti-standard vertex conditions at each vertex. Then, with the help of \textcolor{blue}{Theorem} \eqref{strength}, and the estimate \eqref{anti1}, we obtain 
 $$\lambda_k(\Gamma_a^{\delta'}) \leq \lambda_k(\Gamma^a) \leq \left(\frac{k |E|\pi}{2 L(\Gamma_a^{\delta'})}\right)^2. $$
 
 Let $B$ denote a finite compact connected bipartite metric graph, and let $B^a$ and $B^s$ be the bipartite metric graphs equipped with anti-standard and standard vertex conditions at each vertex of both graphs, respectively. The following relation was established in \cite[\textcolor{blue}{Theorem 3.6}]{PR}.
 \begin{equation}
     \lambda_{k+\beta}(B^a) = \lambda_{k+1}(B^s).
 \end{equation}
 We use this relation and the \textcolor{blue}{Theorem} \eqref{strength} to obtain a relation between the eigenvalues of a  graph with $\delta'$-condition and the same graph equipped with $\delta$-conditions.
 Let $B^{\delta'}$ be the bipartite metric graph and $B^s$ be same underlying bipartite metric graph equipped with standard vertex conditions. Then from \textcolor{blue}{Theorem} \eqref{strength}, we know that either increasing or decreasing the strength at vertex from zero, decreases the eigenvalues. Thus, 
 $$\lambda_{k+\beta}(B^{\delta'}) \leq \lambda_{k+\beta}(B^a)=\lambda_{k+1}(B^s) \leq \lambda_{k+\beta+|V|}(B^{\delta'}).$$
 Let $B^{\delta}$ be a bipartite graph with non-negative strengths at each vertex, then
 $$\lambda_{k+\beta}(B^{\delta'}) \leq \lambda_{k+1}(B^{\delta}).$$
 
Let $|V^p|$ and $|V^n|$ denote the number of vertices in $B^{\delta}$ with positive and negative strengths, respectively, and $B^s$ be the same underlying metric graph. Then the graph $B^{\delta}$ can be obtained from $B^s$ in two steps. First, obtain a graph $B^p$ from $B^s$ by assigning positive strengths to those vertices where positive strengths are assigned in $B^{\delta}$, and then obtain $B^{\delta}$ by assigning negative strengths to some vertices of $B^p$ wherever is required.
Thus 
 $$\lambda_{k+1}(B^{s}) \leq \lambda_{k+1}(B^p)\leq \lambda_{k+1+|V^p|}(B^s) $$
 $$\lambda_{k+1}(B^{\delta}) \leq \lambda_{k+1}( B^p) \leq \lambda_{k+1+|V^n|}(B^{\delta})$$
 $$\lambda_{k+\beta}(B^{\delta'}) \leq \lambda_{k+1}(B^s) \leq \lambda_{k+1}(B^p) \leq \lambda_{k+1+|V^n|}(B^{\delta}).$$
 Therefore, 
 \begin{equation}
    \lambda_{k+\beta}(B^{\delta'}) \leq  \lambda_{k+1+|V^n|}(B^{\delta}) \quad \text{and} \quad \lambda_{k+1}(B^{\delta}) \leq \lambda_{k+\beta+|V^p|}(B^a).
 \end{equation}
Moreover, for any metric graph $\Gamma$, not necessarily bipartite, one can obtain several bipartite metric graphs $B$ by creating vertices of degree two equipped with standard conditions at some interior points of $\Gamma$. Therefore, for any  metric graph $\Gamma^{\delta}$ with non-negative strengths at each vertex, the following relation hold.
\begin{equation}
    \lambda_{k+\beta}(B^{\delta'}) \leq \lambda_{k+\beta}(B^a)=\lambda_{k+1}(B^s)=\lambda_{k+1}(\Gamma^s) \leq \lambda_{k+1}(\Gamma^{\delta}).
\end{equation}
 Note that in the above relation the total length of the graphs $B^{\delta'}$ and $\Gamma^{\delta}$ are equal. That is, $\mathcal{L}(B^{\delta'})=\mathcal{L}(\Gamma^{\delta}).$ 
 Let $\Gamma^{\delta}$ and $\Gamma^{\delta'}$ be same underlying metric graphs with strengths $\alpha_m$ and $\alpha'_m$, respectively.
 Let $\Gamma^{\delta}_f$ be the corresponding flower graph obtained from $\Gamma^{\delta}$, then $\Gamma^{\delta}_f$ can be turned into a $\B^s$ bipartite graph  by introducing the degree two vertices at its leave. Therefore,
 $$\lambda_k(\Gamma^{\delta}) \leq \lambda_k(\Gamma^{\delta}_f) \leq \lambda_{k+1} (\Gamma^s_f)=\lambda_{k+1}(\B^s)=\lambda_{k+\beta}(\B^a)$$
 Let us impose the Dirichlet conditions at degree two vertices of $\B^a$ to obtain a graph $\B^a_d$. Now again impose the standard conditions at the Dirichlet vertices in $\B^a_d$. Thus  
 \begin{equation}
     \lambda_{k+\beta}(\B^a) \leq \lambda_{k+\beta} (\B^a_d) \leq \lambda_{k+2\beta}(\B^a_s)= \lambda_{k+2|E|}(\Gamma^a_f) \leq \lambda_{k+2|E|+1}(\Gamma^{\delta'}_f).
 \end{equation}
 Where $\Gamma^{\delta'}_f$ has strength $\alpha=\sum\limits_{m=1}^{|V|} \alpha'_m$, and
 if for all $m=1,2,\cdots,|V|$, the strengths $\alpha'_m \geq 0$ in $\Gamma^{\delta'}$, then restoring the vertices of $\Gamma^{\delta'}$ from $\Gamma_f^{\delta'}$ implies that 
 $$\lambda_k(\Gamma^{\delta}) \leq \lambda_{k+2|E|+1}(\Gamma^{\delta'}_f) \leq \lambda_{k+2|E|+1}(\Gamma^{\delta'}).$$
 If $|V^{p'}|$ and $|V^{n'}|$ denote the number vertices in $\Gamma^{\delta'}$ with positive and negative strengths, respectively. Then, 
 \begin{align*}
 \lambda_k(\Gamma^{\delta}) \leq \lambda_k(\Gamma^p) &\leq \lambda_{k+|V^p|}(\Gamma^s)\\
 &= \lambda_{k+\beta}(B^s)\\
 &= \lambda_{k+\beta+|V^p|-1} (B^a)\\
 &\leq \lambda_{k+\beta+|V^p|+|V^{p'}|+|V^{n'}|-1}(B^{\delta'})
 \end{align*}
 
 Since for the graph endowed with $\delta'$-type conditions, attaching a pendant edge, attaching an edge between existing vertices, or gluing two vertices with at least one equipped with anti-standard vertex conditions lowers all eigenvalues \cite[\textcolor{blue}{Theorem 3.2, 3.5 and 4.2}]{RS}. Therefore, for any metric graph $\Gamma^{\delta'}$ with length $\mathcal{L}(\Gamma^{\delta'}) \geq \mathcal{L}(\tilde{\Gamma}^{\delta})$, we obtain, 
 $$\lambda_{k+\beta}(\Gamma^{\delta'}) \leq \lambda_{k+\beta}(B^{\delta'}) \leq \lambda_{k+1+|V^n|}(\tilde{\Gamma}^{\delta}).$$
 
%%%%%%%%%%%%%%%%%%%%%%%%%%%%%%%%%%%%%%%%%%%%%%%%%%%%%%%%%%%%%%%%%%%%%%%%%%%%
%\section*{Data Availability Statement}
%The data that supports the findings of this study are available within the %article.

%%%%%%%%%%%%%%%%%%%%%%%%%%%%%%%%%%%%%%%%%%%%%%%%%%%%%%%%%%%%%%%%%%%%%%%%%%

\end{document}